\documentclass[sigconf]{acmart} 

\usepackage{balance} 

\usepackage{enumitem}
\usepackage{amsmath}
\usepackage{appendix}
\usepackage{graphicx}
\usepackage{float}
\usepackage{graphicx}
\usepackage{caption}
\usepackage{subcaption}
\usepackage{color}
\usepackage{url}
\usepackage[textsize=scriptsize]{todonotes}
\usepackage{xcolor}
\definecolor{blue}{RGB}{0, 93, 170}			

\newtheorem{observation}{Observation}

\def\Comments{0} 
\newcommand{\kibitz}[2]{\ifnum\Comments=1{\color{#1}{#2}}\fi}
\newcommand{\rmr}[1]{\kibitz{red}{[RESHEF:#1]}}

\newcommand{\full}[1]{#1}
\def\cite{\citealp}

\title{To Stand on the Shoulders of Giants:  Should We Protect Initial Discoveries in Multi-Agent Exploration? }
\author{Hodaya Lampert, Reshef Meir and Kinneret Teodorescu\\
Technion---Israel Institute of Technology
}

\usepackage{graphicx}
\graphicspath{ {./images/} }

\begin{abstract}
     Exploring new ideas is a fundamental aspect of research and development (R\&D), which often occurs in competitive environments. Most ideas are subsequent, i.e. one idea today leads to more ideas tomorrow. According to one approach, the best way to encourage exploration is by granting protection on discoveries to the first innovator. Correspondingly, only the one who made the first discovery can use the new knowledge and benefit from subsequent discoveries, which in turn should increase the initial motivation to explore. An alternative approach to promote exploration favors the \emph{sharing of knowledge} from discoveries among researchers allowing explorers to use each others' discoveries to develop further knowledge, as in the open-source community. With no protection, all explorers have access to all existing discoveries and new directions are explored faster. 
     
     We present a game theoretic analysis of an abstract research-and-application game which clarifies the expected advantages and disadvantages of the two approaches under full information. We then compare the theoretical predictions with the observed behavior of actual players in the lab who operate under partial information conditions in both worlds.
     
     Our main experimental finding is that the no protection approach leads to \emph{more} investment efforts overall, in contrast to theoretical prediction and common economic wisdom, but in line with a familiar cognitive bias known as `underweighting of rare events'. 
     
\end{abstract}

\begin{document}


\pagestyle{fancy}
\fancyhead{}

\maketitle
\full{
\newpage
\section*{Declarations}
\subsection*{Funding}
This project was funded by the Israeli Ministry of Science and Technology (Grant no. 3-15284) and the Israel Science Foundation (ISF; Grant No. 2539/20).
\subsection*{Conflicts of interest/Competing interests}
Not applicable
\subsection*{Availability of code, data and material}
Data and code can be found in\\ \url{https://github.com/hodayal/}\\
\url{The-Effect-of-Protecting-Initial-Discoveries-on-Exploration}.
 \newpage
 }

\section{Introduction}

Stimulating innovation that leads to advances in technology has always been a core challenge for policy designers. To this end, some proponents advocate free competition while others argue for the benefits of rights protection. Although competition among agents usually has a positive influence on their incentives \citep{nickell1996competition,blundell1999market,younge2018competitive}, competition in the context of innovation might represent an exceptional case. This exception can be attributed to the cumulative nature of discoveries \citep{scotchmer1991standing}: when a new radical discovery is made it paves the path to a whole field of research possibilities. This new knowledge can be used to easily, and cheaply, make many more subsequent, though incremental, discoveries. Thus, if the new knowledge is publicly shared then other inventors can use it to generate incremental discoveries. However, only the first inventor bears the cost of the whole discovery process. Therefore, discoveries and inventions could, arguably, be classified as a public good, and as such might receive insufficient contributions (i.e. exploratory efforts) in a competitive environment \citep{grossman1990trade}.

Existing literature on multiagent search typically focuses on designing better agents that can both cooperate and compete~\citep{yokoo1996multiagent, wray2018integrated}. Alternatively, some papers design reward structures that can be better exploited by existing search algorithms~\citep{hester2010real,biswas2015truthful,jacq2022lazy}. 

In this work, we do not assume we have direct or even indirect control over the players (that may opaque algorithms, or humans, or firms), and would still like to incentivize them to explore better.

One common way to overcome the problem of insufficient exploration is to grant original inventors exclusive rights to explore related incremental discoveries. This type of protection aims at encouraging radical innovation efforts by blocking others from competing on subsequent developments, which in turn increases the potential rewards for original inventors \citep{kaufer2012economics}. \rmr{best if we can provide examples from software/technology} 
For example, technology firms protect their breakthrough discoveries with patents. Patents make it difficult  for other firms to use the protected knowledge and through this action, give the patentees a significant advantage in competition for subsequent products. In academic research, a policy that allows researchers to keep their data private, increases the reward from collecting the data, giving the researcher an advantage over other researchers from the same field (who cannot access and explore the data set).

\full{
Beyond increasing the incentives to search for radical innovations, another potential advantage of protecting initial discoveries is relative specialization. When a research team or a technology firm specializes in investigating one initial discovery, they can learn from their own previous experience which research directions work best and which will fail with high probability. Moreover, providing the original inventors exclusive rights for subsequent searches reduces the chances that the same discovery will be made by several teams in parallel, which in turn increases the overall efficiency of the search process \citep{baron2013cooperates,denicolo2000two}. Accordingly, blocking others from searching for subsequent discoveries may lead to more efficient exploration processes.  

Yet, blocking others from using existing knowledge reduces competition for subsequent discoveries which in fact might slow down the discovery process \citep{llanes2009anticommons,boldrin2005economics,bessen2009sequential,galasso2014patents}. For example, conferences and journals have recently started to condition publication of papers on making the data public \citep{stieglitz2020researchers,zhu2020open}, allowing other researchers to explore the data and extract novel insights as well as find errors in the original studies. This approach already underlies existing open-source platforms, where developers share their source code in a public domain and use other developers' code in their own programs \citep{lerner2006dynamics}. However, notice that open policies might suffer from the disadvantages protection is assumed to solve, i.e. insufficient investment in radical exploration and inefficient exploration processes. In the current work, we aim to shed light on the assumed advantages and disadvantages of protecting initial discoveries. Specifically, we examine how the fundamental economic factors underlying competition with and without a protection policy affect exploration behaviors and performance measurements such as the amount of discoveries made, the speed of making discoveries and the efficiency of the exploration process. 
}

Competitive search for discoveries occurs in many real-life environments, all share the fundamental economic features that are involved in exploration processes. One such prevalent feature is searching costs. 
Search for natural resources, for new innovation or for academic knowledge is costly in terms of both time and money. This cost is heterogeneous among competitors  and also varies over time, due to dynamic environmental factors (e.g. weather, energy price, employee availability). 
\full{Additionally, in many searching processes, successes are public information, while failures are private. Firms and researchers tend to publish their achievements to increase their reputation, their value in the stock market or their profits.\footnote{In some cases firms prefer to keep their successes as trade secrets. This option can be available when granting a patent is too expansive, or legally impossible. However, in many cases the possibility of reverse engineering the final product reduces the effectiveness of this choice.} However, failures such as unsuccessful attempts to find gas or oil, wrong research directions, failed experiments or disappointing development endeavors often remain private information that is kept far from the competitors eyes. 
}
\rmr{I think we can cut short some of this. We mainly need to say why we need both theoretical analysis (for generality) and experiments (because we suspect the way people treat small probabilities)}
Importantly, another shared attribute of competitive innovation environments is the distribution of rewards over the different types of discoveries: initial exploration in unfamiliar areas is less likely to succeed but offers higher rewards for radical discoveries, while subsequent, incremental discoveries are more frequent and yield lower rewards.

The interplay between the magnitude and the frequency of rewards and its effect on behavior cannot be captured by a model focusing on expected utility, but has been extensively explored in the Decisions from Experience (DfE) literature. One of the most robust findings in this literature is that in repeated choice settings, people tend to underweight rare events \citep{barron2003small,hertwig2004decisions,teodorescu2021enforcement}. Specifically, in exploration tasks, participants were found to under-explore in a ``rare treasure environment", where exploration is disappointing most of the time but on rare occasions can yield very high reward (discovery) \citep{teodorescu2014decision,teodorescu2014learned}. 

In the context of innovation, since initial discoveries can be thought of as rare treasures, invested exploration efforts may be below optimum. Moreover, since protecting initial discoveries reduces the probability for others to make subsequent discoveries, it decreases the average probability to make a subsequent discovery \citep{bessen2009sequential}. Thus, underweighting of rare events implies that increasing the magnitude of a rare reward (via protection) will have a smaller than expected effect on exploratory efforts to find initial discoveries. Underweighting of rare events also implies that when exploration is frequently rewarding (i.e. in searches for incremental discoveries) disappointing exploration efforts are more rare and thus people might over-search for subsequent discoveries (searching even when it is not optimal to do so).
\\
\paragraph{Previous experimental studies}
Only a few experimental studies tackled the effect of discovery protection on innovative behavior. \citet{torrance2009patents} used an interactive R\&D simulation, finding that protection reduced both the quantity and quality of innovations, and decreased welfare compared to a no-protection condition. Similarly, \citet{bruggemann2016intellectual} using a Scrabble like creativity task, found that protection reduced innovations' quantity and quality and also reduced welfare. However, \citet{buchanan2014experiment},  using a color generation studio task and \citet{dimmig2012quasi}, using a two-player duopoly game, found no significant or only minor effects of discovery protection. \citet{ullberg2012dynamic,ullberg2017coordination} further highlighted that low patent validity impairs coordination in a licenses market. Importantly, the limited number of experimental studies in competitive environment have employed highly complex tasks, which may increase external validity, but make causal relationships difficult to analyze. For example,  \citet{torrance2009patents} complex simulation does not clarify whether the adverse effect of protection stemmed from patenting cost, probability of making a discovery, licensing availability/fees, or other factors, nor whether participants' behaviors were rational response or influenced by behavioral biases. Additionally, probably due to the complexity, most of the above experiments lasted more than an hour yet included a relatively small amount of trials (10-25 per session). Since participants receive feedback only at the end of each trial, the limited number of trials makes it difficult to address learning and long-term effects.\footnote{\citet{torrance2009patents} are an exception, not employing a distinct-trials setting but rather using a fixed time limit of 25 min.}


\paragraph{The current framework}
In the current study we investigate exploration with and without protection over many trials and with immediate feedback. We aim to shed light on the fundamental causal effect of initial discovery protection on exploration, learning and performance within a competitive sequential environment. To this end, we developed a simplified game in which players compete to find hidden treasures on a spatial map. The competition is sequential, such that exploration decisions are based on existing knowledge that was discovered in previous periods. In this framework, treasures represent successful innovation efforts, i.e. making a new discovery is simulated by finding a treasure. The game is played under two conditions, ``Protection" and ``No Protection". Under the ``Protection" condition, the information gained from a treasure discovery can be used exclusively by the finder,\footnote{Hence, the protection here means that the finder can exclusively enjoy incremental improvement of initial discoveries.} and in the ``No Protection" condition, players can use the location of any treasure to guide the search for subsequent treasures. In addition, in both conditions the players failed exploratory efforts are private information, while their successes are public information. Within this simplified framework, we focus on investigating the effect of protection on exploration for initial and subsequent discoveries as well as on exploration efficiency. 

Unlike some of the previous experimental studies, here we do not focus on the innovation process itself (which involves creativity and entrepreneurship abilities, as in \citet{bruggemann2016intellectual}) but rather on the more basic economic variables such as search costs, the probability to make a discovery, the magnitude of reward obtained following discoveries etc. Importantly, the current setting also allows derivation of proxies to the optimal strategies with and without protection and the comparison of these proxies with actual behavior. Optimal strategy analysis assumes players act rationally and base decisions on full information regarding their payoffs structure. However, given the uncertain nature of innovative activity and evidence for bounded rationality, deviations from optimality might occur, as will be discussed below. Importantly, the current, simple, setting enables identification of systematic behavioral deviations from optimality under full information assumptions, which could be crucial in deriving efficient and ecologically valid policy implementations. 

The rest of the paper proceeds as follows: In Section~\ref{sec:theory} we put forward and analyze an abstract theoretical model of sequential discoveries with and without protection, confirming our main hypothesis that protection encourages initial discoveries but inhibits followup discoveries. In Section~\ref{sec:golddigger} we present a concrete game that simulates such an environment and compare theoretical with behavioral predictions. 
Our main contribution is a large lab experiment in Section \ref{sec:experiment}, in which the theoretical and behavioral predictions were tested, showing some expected and some surprising results. Section \ref{sec:conclusion} summarizes the main results and discusses theoretical and practical implications.

\section{A Theoretical Model for Competitive Exploration}\label{sec:theory}
In this section we put forward an abstract model that allows theoretical analysis. 
\\
There are $n$ players, each of which chooses how much to invest in exploration for novel knowledge (or \emph{research}), and how much to invest in \emph{exploitation} of existing knowledge, that may lean to \emph{application}. 

The strategy of each agent is thus composed of two real numbers, $r_i,x_i\geq 0$, representing the effort  $i$ invests in research and in exploitation of knowledge provided the opportunity, respectively. \rmr{bad terminology. $r_i$ is indeed in terms of effort or work. But $x_i$ is the rate at which $i$ applies available knowledge is consumed/applied}

 We call the aggregated research product \emph{knowledge}, $K:=\sum_{i=1}^n r_i$, which can in turn be exploited for applications. As $x_i$ is the effort $i$ invests in applying knowledge,  the overall \emph{work} $i$ invests in exploiting knowledge is $w_i:=K\cdot x_i$. 

The amount of knowledge $i$ actually applies depends not only on her own exploitation efforts, but also on others', as only one agent can profit from each application. We assume \emph{application} profit $a_i$ is proportional to the exploitation effort and to the total knowledge, so that $a_i:= \frac{x_i}{\sum_j x_j} K$.

\paragraph{Costs and utilities}
Both research and application carry direct benefit to the agent, as well as costs. 

For ease of exposition and consistency with the game we design later, We will associate a fixed reward $R_r, R_a\geq 0$ with each achievement, as well as a single convex \rmr{not sure convexity is needed except for uniqueness. We do need though that the derivative at 0 will be lower than the reward as otherwise no research is performed. } cost function $c:\mathbb R_+ \rightarrow \mathbb R_+$. Convexity of the cost function is due to the decreasing marginal gains of work invested.

We further assume that ceteris paribus, exploitation is more rewarding than exploration per invested effort, and hence $R_a \geq R_r$.

\begin{itemize}
    \item The total knowledge generated is $K:=\sum_i r_i$;
    \item The exploitation work of $i$ is $w_i:=x_i\cdot K$;
    \item The knowledge applied by $i$ is
    $$a_i:=\frac{x_i}{\sum_j x_i}K,$$
    or just $a_i = x_i K$ if there is no competition (i.e. if $\sum_j x_j < K$);
    \item the overall utility of $i$ is
    $$u_i(r,x):=r_i R_r + a_i R_a - c(r_i) - c(w_i).$$
\end{itemize}

    
\paragraph{Protected research}
When initial research is protected (e.g. by patents), there is no interaction between players. In our model, this essentially means that for each player $i$, $K=r_i$. We also replace the index with $0$ to denote it is a single player game. The optimal strategy then becomes a simple optimization problem.

\begin{proposition}
    The optimal strategy in the protected condition is to play $x^*_0 = 1$, and  $r^*_0$ is the unique $r$ s.t. $c'(r)=\frac{R_r+R_a}{2}$.
\end{proposition}

\begin{proof}
If $x_0>1$ then $w_0 > K$, and $a_i=\frac{x_0}{x_0}K=K$. So the agent pays $c(w_0)>c(K)$ without getting any additional benefit beyond $R_a\cdot K$. Thus $x_0>1$ is dominated.

If $x_0<1$ then $a_0=x_0 K = x_0 r_0$, and  
$$u_i=r_0 R_r + a_0 R_a- c(r_0) - c(w_0)= r_0 R_r + r_0 x_0 R_a- c(r_0) - c( r_0 x_0).$$

 We consider both partial derivatives of $u_0$:
\begin{align*}
    \frac{\partial u_0}{\partial r_0} & = R_r+x_0 R_a - c'(r_0)-x_0 c'(r_0 x_0)\\
    \frac{\partial u_0}{\partial x_0} & = r_0 R_a - r_0 c'(r_0 x_0) \tag{since $a_0=x_0 r_0$}\\
\end{align*}
If the strategy is optimal, then both derivatives are 0. However this would mean
$$x_0 \cdot c'(r_0 x_0)= R_r + x_0 R_a -c'(r_0) ; \text{ and }  c'(r_0 x_0)=R_a,$$
and thus
$$x_0 R_a = R_r + x_0 \cdot R_a,$$
 which a contradiction since $R_r>0$.

The strategy of the player therefore reduces to a single variable $r_0$, and the utility can be re-written as $u_0(r_0)=r_0(R_r+R_a) - 2c(r_0)$.
By derivation, we get 
that $r^*_0$ is the unique point where $c'(r) = \frac{R_r+R_a}{2}$.
\end{proof}

\paragraph{No protection}
When there are multiple players with access to the generated knowledge, we have that $K=\sum_i r_i$, and the applications $a_i$ each agent generates depend both on $K$ and the exploitation strategies $x_1,\ldots,x_n$, as explained above. 

\begin{observation}
    In every equilibrium, knowledge is fully exploited. I.e. $\sum_i x_i\geq K$.
\end{observation}
Otherwise, there is an agent with $a_j=x_j<r_j$, and we get a contradiction as in the singleton case.

\begin{proposition}
    There is a symmetric equilibrium, where for every agent $i$, 
    \begin{enumerate}
        \item  $c'(r^*_i)=R_r+\frac{R_a}{n^2}$; and 
        \item $x^*_i c'(n\cdot r^*_i x^*_i)= \frac{n-1}{n^2}R_a$.
    \end{enumerate}
\end{proposition}
For a proof see Appendix~\ref{apx:proofs}.
\begin{corollary}
    The rate of exploration is \emph{higher with protection} as long as $\frac{R_a}{R_r}>\frac{n^2}{n^2-2}$; and the rate of exploiting available knowledge is \emph{lower with protection} as long as  $\frac{R_a}{R_r}>\frac{n^2}{n^2-n-1}$.
\end{corollary}
Note that the condition on $\frac{R_a}{R_r}$ becomes trivial for large $n$. 
\begin{proof}
    For initial search the rate of exploration is just $r$.  Note that since $c$ is convex, $c'$ in increasing and thus $r^*_i > r^*_0$ iff $c'(r^*_i) > c'(r^*_0)$, which means $$R_r+\frac1{n^2}R_a > \frac{R_r+R_a}{2} \iff \frac{R_a}{R_r}>\frac{n^2}{n^2-2}.$$

    For sequential search, note first that the rate at which knowledge is consumed under protection is $x^*_0=1$. Without protection, there is one pool of knowledge of size $K$, which is consumed at rate $\sum_{i=j}^n x^*_j$, i.e. $nx^*_i$ in a symmetric equilibrium. We argue that $x^*_i > \frac1n$ (under the premise assumption on $R_a,R_r$). 

    Indeed, assume towards a contradiction that $x^*_i < \frac1n$. Then due to $c'$ being an increasing function,  
    \begin{align*}
        \frac{n-1}{n^2}R_a &= x^*_i c'(n\cdot r^*_i x^*_i) < \frac1n c'(r^*_i)\\
        &= \frac1n c'( (c')^{-1}(R_r+\frac{1}{n^2}R_a)) = \frac1n (R_r+\frac{1}{n^2}R_a) &\iff \\
        (n^2-n)R_a &< n^2R_r +  R_a &\iff \\
        \frac{R_a}{R_r} &< \frac{n^2}{n^2-n-1}, 
    \end{align*}
    in contrast to out premise assumption. 
\end{proof}

In fact, for polynomial costs we can get an approximate estimate of the actual effort invested in sequential search. Again the proof is in Appendix~\ref{apx:proofs}.
\begin{proposition}
    Suppose that $c(x) = \alpha \cdot x^\beta$.\\
    Then $x^*_i=\frac1n (\frac{R_a}{R_r})^{\frac{1}{\beta}}+\Theta\left(\frac{1}{n^{1+\frac{1}{\beta}}}\right)$.
\end{proposition}
For large $n$, the low order term can be neglected, and we get that the overall rate in which the generated knowledge is exploited is $\sum_j x_j\approx (\frac{R_a}{R_r})^{\frac{1}{\beta}}>1$, i.e. faster than it is under protection. Interestingly, the rate asymptotically depends only on the ratio $\frac{R_a}{R_r}$ and not on the number of the competing agents. 

\section{The Competitive Treasure Hunt Game}\label{sec:golddigger}

``The Competitive Treasure Hunt" game is played in groups of $n=4$ players. In this game, players are faced with a hive of white hexagons and need to find treasures. 5\% of the hexagons are hidden treasures that simulate discoveries in the real world. 

 Every three treasures are arranged in clusters  which form a tight triangle. We define the three linked treasures as a ``gold mine." Therefore, discovering one treasure increases the probability of finding the second treasure in the mine from (roughly) 0.05 to at least 0.33. 
 The value of the first treasure in the cluster is set to 320, so the expected reward of every `research' action is $R_r = 0.05\cdot 320=16$.  The value of subsequent treasures is only 80, so we can think of the expected reward as (at least) $R_a = 0.33\cdot 80 \approx 26.6$, and in particular higher than the reward for initial research.

 The first treasure to be found in each mine simulates a breakthrough discovery and the other two treasures simulate sequential discoveries. While finding an initial innovation is rarer, it provides knowledge that increases the probability of sequential innovations, or in our game, subsequent treasure discoveries. 
 

 The costs of exploration for each round are uniformly distributed  over $\{5,10,15,20,25,30,35\}$, and are sampled independently for every player in each round. Each player is informed of his current cost of exploration at the beginning of each round.\footnote{The variation in search costs is intended to create heterogeneity between the players that also exists in the real world, where sometimes certain players have more skill (or knowledge, or resources) that allows them lower cost compared to others.} The players choose simultaneously whether to explore or to skip the round. Players who decide to skip the round obtain 0, and players who decide to explore, get to search one of the hexagons in the hive. They must pay the costs of exploration, and their total payoff in the round depends on whether they find a treasure or not, under which condition they play, and the decisions of the other players.\footnote{The reason we chose this payoffs and cost structures is because we designed the optimal search cost threshold strategies to be roughly in the middle of the cost range, to reduce ceiling or floor effects. The calculation of the optimal strategies can be found in the chapter of the theoretical analysis.}
 Mapping costs to our theoretical model, we get that $c$ is roughly quadratic. To see why, suppose that search costs were uniform in $[0,35]$ rather than discrete, then an agent searching whenever the cost is under some threshold $t$ would end up paying $\int_{\ell=0}^t t dt =\frac{t^2}{2}$.

After clicking on a hexagon, if a player does not find a treasure, the hexagon he choose is colored in black on his board, but not on the other players' boards. If a player finds a treasure, the hexagon is colored in yellow on his board, and in red on the other players' boards (thus treasures are public, but failed exploration efforts are not).  

After Once a hexagon is colored in any color, the player cannot choose this hexagon in future rounds of that game. The mines are not adjacent to each other. Also, the treasure map was built so that all the mines contained exactly 3 treasures.\footnote{Regarding the edges, the proportion between treasures and empty hexagons approximately remains, so that the probabilities to find a treasure were not affected by the mine's location.}   
 The game is played 4 times with 50 rounds each. The objective of the game is to maximize the expected payoff in each round.

The game is played under two conditions: ``Protection" and ``No Protection". 
\paragraph{Protection condition}
Under the ``Protection" condition, whenever a player finds the first treasure in a new mine, he also obtains the exclusive right to explore the adjacent hexagons (note that this area covers the entire gold mine). The protected area is marked on the board for all players, and   the marking is removed once the entire mine was discovered (see Fig. 1). Hence, no other player can profit from the information revealed after finding the first treasure in a new mine, since collecting the payoff from the two other treasures is not possible. 

\rmr{move to appendix: In addition, when two or more players find the same first treasure simultaneously, they all pay the costs of exploration, and the computer randomly chooses one of them to receive the gains of the payoff from the treasure and the protection of the mine, while the others obtain 0 and cannot profit from subsequent discoveries in that mine. 
}

When a protected treasure is discovered, the protection allows the player to profit exclusively from all hexes adjacent to the treasure. A protection boundary is created that signals to the player with the protection and to the other players that there is an active protection. The protection boundary continues to be marked until all the treasures in the mine have been discovered. 

\paragraph{No Protection condition}
Under the ``No Protection" condition, when a player finds a treasure, this does not restrict the future search of other players. 

After choosing a hexagon, it is colored as in the case of the Protection condition.

See Figures~\ref{fig:screenshotpatent} and \ref{fig:screenshotnopatent} for screenshot examples. E.g. in Fig.\ref{fig:screenshotpatent}  we can see some failed searches, one mine that was fully discovered by the current player, and two mines that are partially discovered: one protected by the current player (with a single discovered treasure); and one protected by another player (with two treasures discovered out of three).  In Fig. 2 we can see two fully discovered mines, where the current player managed to obtain some of the profit.\footnote{Examples of screenshots of typical end games in both conditions are presented on Appendix~\ref{apdx:screenshots}}

\begin{figure}[t]
\includegraphics[width=0.45\textwidth]{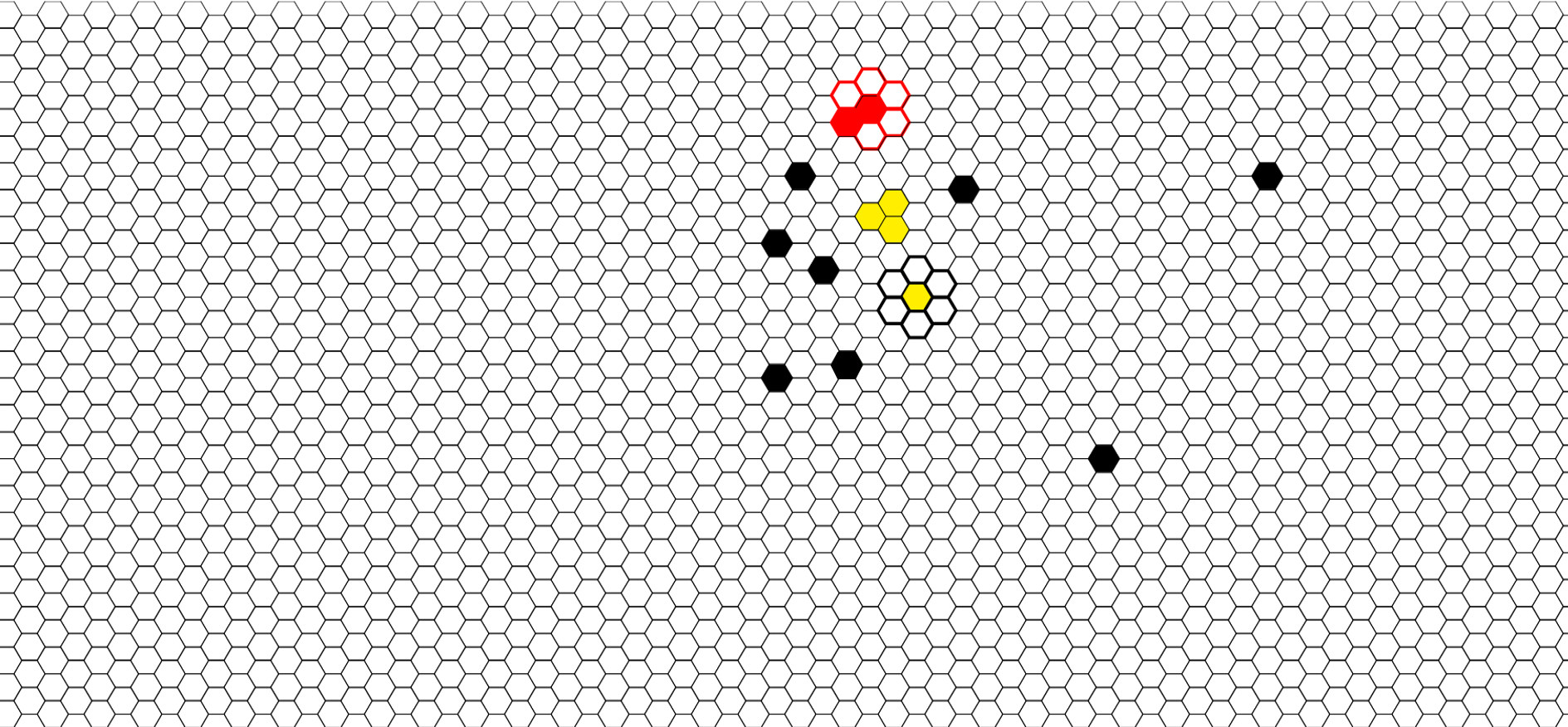}
\centering
\caption{A screenshot of the game, the Protection condition. Black hexagons represent failed searches, red hexagons are treasures that were found by other players and yellow hexagons are hexagons that were found by the player himself.}
\label{fig:screenshotpatent}
\end{figure}
 \begin{figure}[t]
\includegraphics[width=0.45\textwidth]{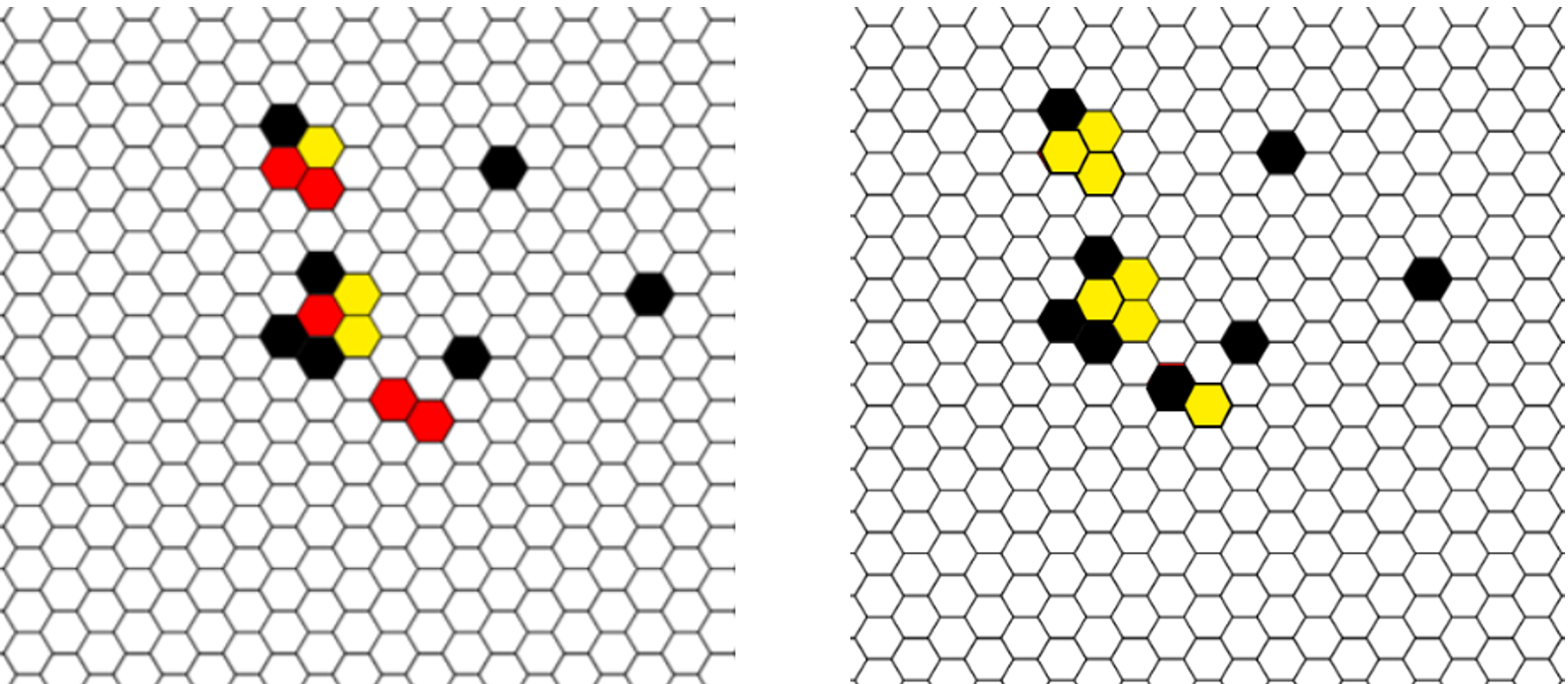}
\centering
\caption{Left: Part of a screenshot in the No Protection condition. Notice that in this condition, there are no protected areas thus each mine can be discovered by more than one player. Right: Part of a screenshot in the Singleton condition.}
\label{fig:screenshotnopatent}
\end{figure}

In addition, we refer to a control ``Singleton"  condition.\footnote{A similar approach was taken by~\citet{levy2018understanding} in a different setting where the researchers study the effect of competition on the players' behavior in simple contests.}


Under the ``Singleton" condition, each player plays as a singleton player, completely unaffected by other players  (his payoff and his board is independent of the other players' choices).  
The players can only observe their own treasures (colored yellow) and failed exploratory efforts (colored black).

\subsection{Simulation Results}\label{sec:simulations} 
We programmed artificial Fully Informed Baysian Players (FIBP) \rmr{what is FIBP?} in both conditions, and let them play ``The Competitive Treasure Hunt" game, in order to  provide a theoretical prediction regarding the players' performance in the game. A player is defined by a pair of thresholds: cost thresholds for exploration for first and for subsequent treasures. We focus on symmetric strategies, i.e. within each simulation, all players use the same combination of strategies. The number of treasures found and their payoffs were the outputs. We repeated the game 10,000 times for each possible combination of thresholds. 

The simulation results qualitatively confirm the results of the abstract model with the appropriate parameters set. In particular, from the results in the previous section:
\begin{itemize}
    \item $r_0^* = \frac{1}{4\alpha}(R_r+R_a) = \frac{21}{2\alpha}$;
    \item $r_i^* = \frac{1}{2\alpha}(R_r+\frac{1}{n^2}R_a) = \frac{17.6}{2\alpha}$;
    \end{itemize}
    so we would expect an increase of $\sim 20\%$ in initial search frequency when adding protection.  
    
    Likewise, since $\left(\frac{R_a}{R_r}\right)^\frac1\beta=\sqrt{\frac{26}{16.6}}\approx 1.25$, then for large $n$ we should expect an decrease of $\sim20\%$ in the rate of sequential search under protection,  although when considering the low order terms for $n=4$ we get a much smaller expected decrease of about $5\%$.

Indeed, in our simulations the optimal/equilibrium initial search threshold increases from 15 to 20 when applying protection, and sequential search threshold decreased from 25 to 20. We should note however that the simulation only used multiples of 5 so it is not very precise.  \rmr{If Hodaya still has the code I'm sure she can run it with higher threshold resolution}
 
See more details on the simulations in Appendix~\ref{apdx:simulation}. 

\full{The simulations show that under a rational behavior assumption FIBP find more treasures in the Protection condition, both at the group and at the individual levels. In addition, the simulations show that the number of treasures found by more than one player, which indicates the inefficiency of exploration, increases as the thresholds increase in the No Protection condition, since higher thresholds result in more search activity over a limited area. The results of this analysis provide the theoretical prediction that under profit maximizing assumption, exploration will be more efficient (less effort leads to more discoveries) under the Protection conditionizing assumption, exploration will be more efficient (less effort leads to more discoveries) under the Protection condition.
 }

\subsection{Theoretical Predictions}
Following the theoretical analysis and simulations, we get two clear theoretical predictions under profit maximization and full information assumptions:

\textit{\textbf{Theoretical Prediction~1:}} Under the Protection condition, initial and sequential search activities should be at a similar rate.

\textit{\textbf{Theoretical Prediction~2:}} Protection increases exploration activity for first treasures.

\textit{\textbf{Theoretical Prediction~3:}} Protection decreases exploration activity for subsequent treasures. \rmr{With the corrected analysis this seems to depend on the game parameters. What we can predict conclusively is that without protection we should expect (slightly) more effort on sequential search than for initial search. This is in agreement with BP~1. So perhaps we should name the theoretical predictions 1 to 1.1 and 3 to 1.2.}


Both predictions stem directly from the theoretical analysis and are supported by simulating rational behavior. 

\subsection{Behavioral predictions} \rmr{Is this the right place to discuss?}
It is important to note that the theoretical predictions were derived under the assumption of FIBP who know the a-priori probability of finding a treasure in a new mine and follow the optimal exploration threshold from start. In real life however (and also in our lab experiment), the a-priori probability of making a new discovery is unknown to the competing players in advance, and they can learn it only throughout ongoing experiences. Under such partial information conditions, it might take time  until rational players converge to a consistent exploration threshold. Importantly, assuming effective learning processes, the consistent exploration threshold which is eventually formed should still be close to the optimal one.

However if participants are evaluating the value of an action based on the \emph{likelihood} for profit more than on the \emph{magnitude} of profit, (in line with underweighting of rare events findings in DfE, e.g.,  \citet{barron2003small,hertwig2004decisions,camilleri2011and,erev2014maximization,plonsky2015reliance,teodorescu2021frequency}) this would alter our hypotheses. First, regardless of the condition (with/without protection), we would expect sequential search to be much more lucrative as it is much more likely to yield a reward. Second, while protection increases the magnitude of initial rewards, it does not affect the chance of success, and thus should not have a major effect on initial search behavior. This yields the following predictions:

\textit{\textbf{Behavioral Prediction~1:}} Sequential search activity should be higher than initial search activity, under both conditions.

\textit{\textbf{Behavioral Prediction~2:}} we should not expect a difference in initial search activity between the two conditions.

Note that each of these behavioral predictions~1,2 directly contradicts its theoretical counterpart, while theoretical prediction~3 is not affected by the above discussion.



\section{Experiment Design}\label{sec:experiment}
\paragraph{Participants}
We had 154 subjects divided into groups of 4. In total we had 15 groups in each condition, plus 34 subjects who played in the singleton condition. 
Subjects were payed a show-up fee plus a performance fee that could be either positive or negative.

\paragraph{Experimental Design}
Participants played a lab adaptation of the ``The Competitive Treasure Hunt" game that was described in Section~\ref{sec:golddigger}. Each player played four sessions (with the same group), where each session lasted 50 rounds. We excluded the last 12 rounds of every session from analysis to avoid endgame effects.

For each game, the computer chose randomly one of ten different possible ``maps" of treasures. 

Each player can observe the other players' successes but cannot observe their failures. 
The exploration costs were randomized between participants and between rounds. 

Players were not informed in advance of the probability to find a treasure, yet the number of rounds was sufficient to quickly learn the probabilities. Indeed, analysis of potential changes throughout the games revealed quick learning and no significant differences between the first and the last game participants played.


 See Appendix~\ref{apx:experiment} for further details on participants and experiment design. 
 
\paragraph{Game flow}
 At each trial's onset, the players is informed about the exploration cost for this round, and is asked whether he wants to skip the round or to explore under the current exploration cost.

If a player decided to skip, he gained 0, and the round was over for him. If a player decided to explore, he could choose one of the hexagons in the hive that was not yet colored. At the end of each round, the players received a message with their payoff from the round, calculated as the reward obtained minus the current exploration cost.   

\paragraph{Data Analysis}
We compare the rate of  `search' decisions between the three conditions and between initial and sequential search contexts. While it is straightforward to classify subjects to conditions, we should be more careful when determining the context. 

To avoid ambiguity and maintain consistency among conditions, we considered as `initial search' context all turns in which there were no partially-discovered mines (i.e. mines with one or two discovered hexagons) on the board. \rmr{please verify this!}

We considered as `sequential search' context for a player all turns immediately after discovering a mine. All other turns were excluded from analysis. 

We are performing three types of analysis:
\begin{itemize}
    \item For each such combination of condition and context, we consider the fraction of turns in which the agent chose to search, which we can plot on a curve.
    \item We use linear regression on each context separately, to test for statistical significance of the effect of condition (see details in Appendix~\ref{apx:stat}).
    \item We identify the cost threshold of individual participants, and compare the distribution of thresholds between conditions and between contexts.
\end{itemize}
The threshold analysis is more challenging as participants do not always make consistent decisions (see Appendix~\ref{apdx:actualthreshold} for details). Yet it has the added benefit that we can compare the numerical values we obtain to our theoretical predictions.

\section{Experimental Results}

\subsection{Initial vs. Sequential search}
Recall that Theoretical Prediction~1 suggested no difference between search frequency under the protection condition. Our empirical results show that the search rate for sequential discoveries was $0.72$ vs. $0.6$ for initial discoveries, i.e. an increase of 20\%.  This finding is consistent over all search costs. 

We can consider the same question by comparing participants' search thresholds. The median threshold was higher than the theoretically optimal threshold of 21 \rmr{fix when we have accurate simulations} in both contexts.  Yet the difference was slight for initial search (median threshold of about 22), and substantial for sequential search (median threshold of about 26).  Results in the singleton condition were very similar. 

We can therefore decisively reject Theoretical Prediction~1 in favor of the competing Behavioral prediction. 

\subsection{The effect of protection on initial search behavior}
Focusing on initial search, we do observe some difference in search rate between conditions. There is a decrease of $\sim 11\%$ (from $0.6$ to $0.53$). This decrease is consistent over search costs (see Fig.~\ref{fig:first}), but the effect of the condition is not statistically significant. 

Considering thresholds we get a similar picture: the median threshold is somewhat lower in the No Protection condition (19 vs. 22), yet  higher than the theoretical equilibrium threshold of 17.6.

We therefore see partial evidence both for Theoretical Prediction~2 and to its Behavioral counterpart: Protection somewhat increases initial search activity, but the underweighting of rare events partially mitigating this positive effect.

To corroborate whether this is indeed due to underweighting, we ran the Singleton condition (which theoretically identical to the Protection condition) 
using two levels of reward for a first treasure: 320 (which is the same as the reward for first treasures in the Protection and No Protection conditions) and 260. The results show no significant effect of the reward level on the exploration rates, supporting the interpretation of our result above, whereby players are under-sensitive to the reward level (see supporting evidence in Appendix~\ref{apdx:levelsingleton}).

This result is in line with studies in the DfE literature (e.g., \citet{teodorescu2014learned,teodorescu2014decision}) which found that the frequency of a reward is more important than its magnitude in repeated settings where the environment is learned from experience.

\begin{figure}[t]
  \centering
  \begin{minipage}[t]{0.45\textwidth}  \includegraphics[width=\textwidth]{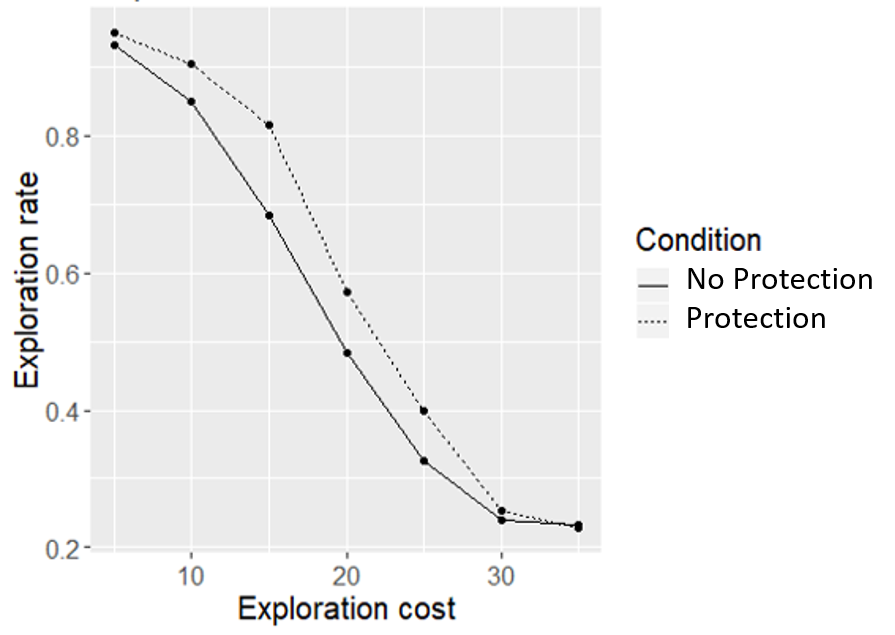}
    \caption{Exploration rates for first treasures.}
    \label{fig:first}
  \end{minipage}
  \hfill
  \begin{minipage}[t]{0.45\textwidth}
 \includegraphics[width=\textwidth]{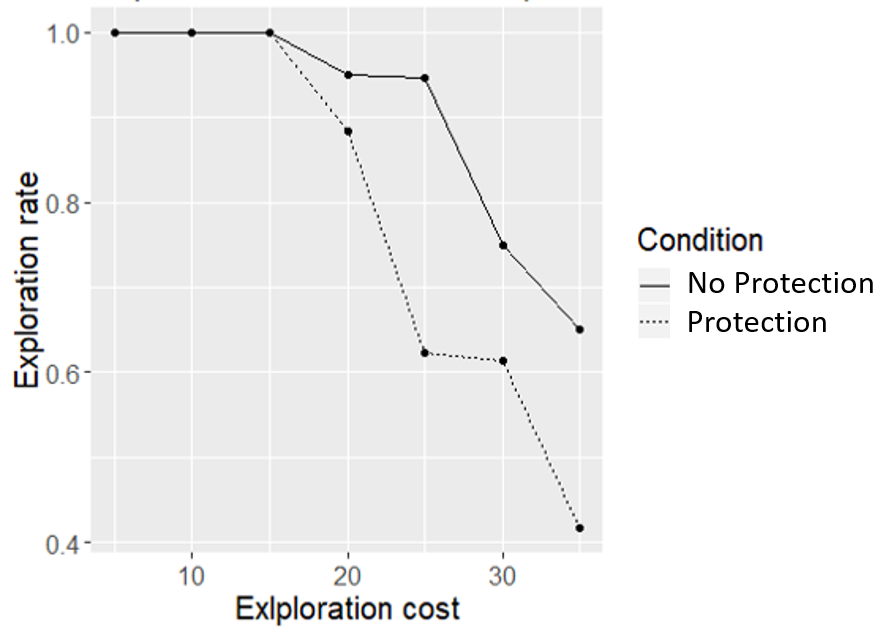}
    \caption{Exploration rates for subsequent treasures.}
    \label{fig:ceiling}
  \end{minipage}
\end{figure}

\subsection{The effect of protection on sequential search behavior}
Finally, we considered the effect of protection on sequential search behavior. 
We see an overall  decrease of 13\% from ($0.9$ to $0.78$)\rmr{check}. A more careful regression analysis reveals the effect as slightly smaller (10 percentage points, which are about $11\%$) and statistically significant. This may still not seem like a substantial difference, but Fig.~\ref{fig:ceiling} reveals the reason: there is a ceiling effect with subjects always searching when costs are $\leq 15$ under both conditions. When restricting analysis to turns with costs above 15, we observe a stronger effect of 17 percentage points (19\%), which is still statistically significant. 

As with previous results, this is corroborated by our threshold analysis, with the median threshold dropping from 29 to 26 when applying protection.\footnote{Here too there is a ceiling effect with many participants who \emph{always} search for sequential rewards.} 
We therefore find conclusive evidence for Theoretical Prediction~3. 

\subsection{Search Efficiency}
The results above imply a strong evidence against protection of initial discoveries, since its potential benefit for initial search is diminishing, while they substantially harm sequential search. 

However this point of view only considers the search efforts invested by participants, rather than its actual yield. 

We therefore ran another analysis, this time considering the actual overall number of treasures found under each condition, divided by the overall search costs. 


When comparing the number of searches to treasures found, we see a sharp drop from 8.4 searches-per-treasure without protection, to about 7 one protection is applied. 
This is due to players `wasting' some of their searches on treasures eventually picked by others. This occurs both when two or more players simultaneously dig the same treasure, and when players search around the same first initial treasure, all exerting effort  but competing for  the only two available treasures. This latter inefficiency (but not the former) is also captured in our abstract model: in the initial search the rate of finding treasures is 1 (as $K=\sum_i r_i$) so there is no inefficiency. In sequential search the agents exert a total effort of $\sum_i w_i = (\sum_i x_i)K$, but only get $\sum_i a_i = K$ applications, which implies that a fraction of $\sum_i x_i -1$ of the (sequential) search effort is wasted.

Thus, regardless of its effect on search behavior, protection has the added benefit of \emph{coordinating} players' effort.

\paragraph{Discussion:}  \rmr{This can be much shorter}Granting protection on a first treasure plays the role of marking the territory for the first finder, and signals other players to explore elsewhere. By doing so, it increases the coordination among the players. This allows the first finder to invest their exploration efforts more carefully, and to search more efficiently by maximizing the information gain from their successes as well as from their failures.    

\citet{bessen2009sequential} discussed complementary research and its effect on the patent protection efficiency. They defined complementary research as a case where
\begin{quote}
    ``each potential innovator takes a different research line and thereby enhances the overall probability that a particular goal is reached within a given
time"
\end{quote}
 They claimed that in the case of sequential and complementary research, the inventor and the society would be better off with no patent protection since
 \begin{quote}
    ``it helps the imitator develop further inventions and because the imitator may have valuable ideas not available to the original discoverer"
\end{quote}
In our design, this is reflected in sharing searching opportunities with players who have lower current costs, or more information regarding the subsequent treasures' location.  


\rmr{to discussion:
In the game, sequential treasures are organized in a specific structure and all players know the possible ways to find them. In this case less search activity means slower discovery process. In reality, however, subsequent discoveries are not a certain result. Therefore, sequential discoveries may be completely missed altogether due to the reduction in search efforts.
}




\subsection{The Effect of Observing Others' Success}\label{sec:forgone}
Finally, we wanted to see if the fact that participants played in a group affected their behavior, even when they are not competing (i.e. in the Protection condition). 
Theoretically, the Protection and Singleton conditions are the same, and we can therefore attribute any differences to behavioral factors, and in particular the fact that  in the singleton game players cannot see the other players' actions.

Indeed, we found that players search a bit more for initial treasures in the Protection condition (i.e. when observing others) but this is not statistically significant. 

In contrast, the increase in search rate immediately after another player finds a treasure was significant,  
These results are in line with \citep{plonsky2020influence} who found that exposing participants to a positive forgone payoff, increases risky exploration. 

Further details can be found in Appendix~\ref{apdx:forgone}.

\section{Conclusions and General Discussion}\label{sec:conclusion}
In this study we developed a new paradigm to investigate the question of the effectiveness of protecting discoveries as a tool to encourage innovation. This paradigm distinguishes between first treasures, that represent initial discoveries, and subsequent treasures, that represent subsequent discoveries. First treasures are found with low probability, and require no previous information. Subsequent treasures are found with higher probability and rely on information obtained from first treasures. 

Our findings show that the benefit of protecting subsequent searches around initial discoveries stems from increasing exploration efficiency, rather than encouraging exploration intensity. While the theoretical benchmark analysis imply that protection should increase the search for first treasures, the observed exploration rates for first treasures did not differ significantly between the experimental conditions. This result is in line with behavioral studies that found that when probabilities are learnt through on going feedback (like in the current experiment and in most real life scenarios), people tend to underweight rare events \citep{ teodorescu2014learned, teodorescu2014decision}. Since finding a first treasure is a rare event, players were under-sensitive to the reward it yields, thus increasing this reward through protecting first treasures did not enhance exploration activity as theoretically predicted.

Furthermore, discovery protection decreased exploration for subsequent treasures. This result is in line with previous studies regarding patent protection (e.g. \citep{bessen2009sequential,boldrin2008against,galasso2014patents}) and attests to the negative aspects of discovery protection i.e. the inhibition of cumulative innovative activity. 

We found that the main advantage of such protection is by improving coordination among players and thereby increasing exploration efficiency. Introducing protection forces a wider distribution of exploration efforts, and reduces duplicated searches. It also allows more efficient exploitation of knowledge about failed exploratory efforts. Hence, other knowledge management mechanisms that maintain these benefits of discovery protection, but without inhibiting competition for innovations, should be considered.

One suggestion can be to encourage communication between explorers about unsuccessful searches. In the case of protecting discoveries, this kind of information can be obtained by employing a market for failed R\&D efforts. Currently, firms' research failures are treated as trade secrets and withholding this information from other researchers leads to an inefficient allocation of exploration efforts. Since information about failure is valuable, allowing a trading mechanism may cause Pareto improvements in the R\&D market.

Another possible alternative to the protection system is to implement an insurance mechanism in the innovation market. Insurance companies could compensate innovators for failed exploratory efforts, and charge a share of their profit from successful exploration. Insurance companies will have an incentive to reveal information about failed exploratory efforts to minimize paying compensation costs to other inventors exploring the same direction. This would improve coordination among innovators, without the social cost of legal monopoly.

In the case of innovating behavior within organizations, researchers can improve their coordination by forming a dedicated forum where they discuss their failures and learn from them. Researchers often tend to cover their failures to avoid bad reputation. Therefore, managers should motivate their employees to share their failures  and draw conclusions for future trials. 

A different policy implication we derive from our findings is that the disclosure of discoveries plays an important role in encouraging innovation. When an inventor observes a successful discovery made by a rival inventor, it encourages him to explore more intensively. Protecting discoveries by granting the researcher who made them an exclusive right to search for subsequent discoveries may increase search efficiency (as our results demonstrate) but protecting discoveries in the sense of keeping their very existence a secret (such as trade secrets) can lower the motivation of others to explore .

Last, it is important to note that investigating the effect of protection through a simplified game setting, bears some limitations. While it allows for collection of more tractable data, it may decrease ecological validity. For instance, creativity and inspiration could not be considered in such reduced form. In addition, the simple setup limits the scope of discussion to the difference between two specific boundary regimes, with and without protection, despite the fact that most of real-world situations lay on a spectrum between these two extreme regimes. Moreover, our setting excludes cases where the initial discoveries are worth less than its improvements. For example, mRNA vaccines have been around before Covid, but it looks like the adaptation of the technology to Covid was financially more lucrative. We also assumed similar costs distribution of all players, and of all treasures, where in reality this assumption may not always hold. Finally, in order to keep the game as simple as possible, we did not include the option to sell and buy licenses in the "Protection" condition. Licensing allows innovators to sell their rights to other innovators, that may have lower searching costs, and by that to improve efficiency \citep{phelps2018need}. However, licensing is not a guarantee to knowledge transfer due to problems as transaction cost, partial information and other market failures \citep{galasso2014patents}.

Future extensions of the current theoretical and experimental work could be to explore the optimal length (e.g. the number of rounds) or scope (e.g. the number of protected hexagons) of protecting initial discoveries. Manipulating heterogeneity among players (e.g. in the cost distribution or in the scope of the searching area) can also provide an interesting insight. Finally, trade in search licensing may improve the ecological validity of the experimental paradigm, shedding light on the (in)efficiency of licensing policies.

\bibliographystyle{ACM-Reference-Format} 
\bibliography{patent.bib}
\clearpage
\onecolumn
\appendix
\section{Omitted Proofs}\label{apx:proofs}

\rmr{
\begin{proposition} Consider general $c_r,c_a$ cost functions. 
    The optimal strategy in the protected condition is to play $x^*_0 = 1$, and  $r^*_0$ is the unique $r$ s.t. $c'(r)=\frac{R_r+R_a}{2}$.
\end{proposition}

\begin{proof}
If $x_0>1$ then $w_0 > K$, and $a_i=\frac{x_0}{x_0}K=K$. So the agent pays $c(w_0)>c(K)$ without getting any additional benefit beyond $R_a\cdot K$. Thus $x_0>1$ is dominated.

If $x_0<1$ then $a_0=x_0 K = x_0 r_0$, and  
$$u_i=r_0 R_r + a_0 R_a- c_r(r_0) - c_a(w_0)= r_0 R_r + r_0 x_0 R_a- c_r(r_0) - c_a( r_0 x_0).$$

 We consider both partial derivatives of $u_0$:
\begin{align*}
    \frac{\partial u_0}{\partial r_0} & = R_r(1+x_0) - c'_r(r_0)-x_0 c'_a(r_0 x_0)\\
    \frac{\partial u_0}{\partial x_0} & = r_0 R_a - r_0 c'_a(r_0 x_0) \tag{since $a_0=x_0 r_0$}\\
    &\geq R_r - c'_a(r_0 x_0) \tag{by assumption} \\
    &> R_r - c'_a(r_0) \tag{by convexity and $x_0<1$}
     =  \frac{\partial u_0}{\partial r_0}.
\end{align*}
If the strategy is optimal, then both derivatives are 0. However this would mean
\begin{align*}
    R_r &< R_a=c'(r_0 x_0) < c'(r_0) &\Rightarrow\\
    \frac{\partial u_0}{\partial r_0} & = R_r+x_0 R_r - c'(r_0)-x_0 c'(r_0 x_0)\\
    &< c'(r_0) + x_0 c'(r_0 x_0) - c'(r_0)-x_0 c'(r_0 x_0)=0,
\end{align*}
i.e. a contradiction.

The strategy of the player therefore reduces to a single variable $r_0$, and the utility can be re-written as $u_0(r_0)=r_0(R_r+R_a) - 2c(r_0)$.
By derivation, we get 
that $r^*_0$ is the unique point where $c'_r(r)+c'_a(r) = R_r+R_a$.
\end{proof}
}
\begin{proposition}
    There is a symmetric equilibrium in the no-protection model, where for every agent $i$, 
    \begin{enumerate}
        \item  $c'(r^*_i)=R_r+\frac{R_a}{n^2}$; and 
        \item $x^*_i c'(n\cdot r^*_i x^*_i)= \frac{n-1}{n^2}R_a$.
    \end{enumerate}
\end{proposition}
\begin{proof}
\begin{description}
    \item[Exploration] Suppose that exploitation strategies $x_j$ are fixed, and by symmetry they are all equal so $x_j=x$ for some constant $x$.  We look for equilibrium exploration efforts $r_i$.
    Rewriting the utility function as a function of $r_i$,
    \begin{align*}
        u_i(r_i)&=r_i R_r + a_i R_a - c(r_i) - c(K x_i)
    \end{align*}
    \rmr{somehow the proof is completely wrong but the result is correct...
    The derivatives in equilibrium are 
    \begin{align}
    0&=\frac{\partial u_i}{\partial r_i} = R_r+\frac1n R_a - c'_r(r_i)-x\cdot c'_a(\sum_j r_j x) \label{eq:pu_r}\\
    0&=\frac{\partial u_i}{\partial x_i} = KR_a \cdot \frac{x_{-i}}{(x_i+x_{-i})^2}-K\cdot c'_a(K x_i) \label{eq:pu_x}
    \end{align}
    Now let $r^*,x^*$ be the symmetric equilibrium strategies, then $x^*_{-i}=(n-1)x^*$. We get
    \begin{align*}
        T^*&:= x^*\cdot c'_a(\sum_j r^*_j x^*) = x^* \cdot c'(nr^* x^*)\\
        T^*&= R_r+\frac1n R_a -c'_r(r^*)   \tag{From Eq.~\eqref{eq:pu_r}}\\
        T^*&= \frac{n-1}{n^2}R_a &\Rightarrow  \tag{From Eq.~\eqref{eq:pu_x}}\\
        &R_r+\frac1n R_a-c'_r(r^*)   = \frac{n-1}{n^2}R_a &\Rightarrow\\
        c'_r(r^*)&= R_r+\frac1n R_a-\frac{n-1}{n^2}R_a  = R_r+ \frac{1}{n^2}R_a
    \end{align*}
    }
    
    Recall that $x_i\geq k_i$ and thus 
    $$a_i=k_i = K\cdot \frac{x_i}{\sum_j x_j}= K/n = \frac{\sum_j r_j}{n}.$$
    Thus the utility as a function of $r_i$ is 
    $$u_i(r_i)=r_iR_r + \frac{R_a}{n}(r_i + r_{-i})-c(r_i)-c(x_i),$$
    and the derivative is 
    \begin{align*}
        \frac{\partial u_i}{\partial r_i}&= R_r + \frac{R_a}{n}-c'(r_i)\\ &\Rightarrow
        r^*_i = (c')^{-1}(R_r + \frac{R_a}{n}).
    \end{align*}
    \item[Exploitation]  Now assume that exploration strategies $r_j$ are fixed, so $K=\sum_j r_j$ is a constant, and we look for $x_i$. By our assumption of over-exploitation, $a_i=K\frac{x_i}{x_i+x_{-i}}$, and 
    $$u_i(x_i) = r_iR_r + a_i R_a -c(r_i)-c(K\cdot x_i) = r_iR_r + R_a\frac{x_i}{x_i+x_{-i}}K  -c(r_i)-c(K\cdot x_i).$$
    Taking derivative, 
    \begin{align*}
          0=\frac{\partial u_i}{\partial x_i} & = KR_a \cdot \frac{x_{-i}}{(x_i+x_{-i})^2} - K\cdot c'(K\cdot x_i)&\Rightarrow \tag{in eq.}\\
          x_{-i} R_a &= c'(K\cdot x_i)(x_i+x_{-i})^2\tag{assuming symmetry}\\
           (n-1)x KR_a&= c'(K\cdot x)(nx)^2 &\Rightarrow\\
           c'(K \cdot x)x=&R_a \frac{n-1}{n^2}, 
    \end{align*}
    which profs the theorem as $K=nr^*_j$.
\end{description}
\end{proof}

\begin{proposition}
    Suppose that $c(x) = \alpha \cdot x^\beta$.\\
    Then $x^*_i=\frac1n (\frac{R_a}{R_r})^{\frac{1}{\beta}}+\Theta\left(\frac{1}{n^{1+\frac{1}{\beta}}}\right)$.
\end{proposition}
\begin{proof}
    First note that $c'(x)=\alpha\beta x^{\beta-1}$. 
    Recall that $c'(r^*_i)=R_r+\frac{1}{n^2}R_a$ so that 
    $$r^*_i = (\frac{R_r+\frac{1}{n^2}R_a}{\alpha\beta})^{\frac{1}{\beta-1}}
    $$
    Next, recall that 
    $$x^*_i c'(K x^*_i) = \frac{n-1}{n^2}R_a = \frac{1}{n}R_a+ \Theta(\frac{1}{n^2}),$$
    where $K=nr^*_i = n(\frac{R_r+\frac{1}{n^2}R_a}{\alpha\beta})^{\frac{1}{\beta-1}} $.
    Plugging in our $c'$, 
    \begin{align*}
        \frac{1}{n}R_a& = x^*_i \alpha\beta \left(  n(\frac{R_r+\frac{1}{n^2}R_a}{\alpha\beta})^{\frac{1}{\beta-1}}\cdot x^*_i\right)^{\beta-1} + \Theta(\frac{1}{n^2})\\ 
        &=(x^*_i)^\beta\cdot \alpha\beta \cdot n^{\beta-1}\frac{R_r+\frac{1}{n^2}R_a}{\alpha\beta} + \Theta(\frac{1}{n^2})&\Rightarrow\\
        R_a &= (x^*_i)^\beta\cdot n^{\beta}(R_r+\frac{1}{n^2}R_a) + \Theta(\frac{1}{n}),
    \end{align*}
    so we can already see that $x^*_i=\Theta(\frac1n)$.
    Thus we can continue
    \begin{align*}
               R_a &= (x^*_i)^\beta n^{\beta}R_r+ (x^*_i)^\beta n^{\beta}\frac{1}{n^2}R_a + \Theta(\frac{1}{n}) \\
               &=(x^*_i)^\beta n^{\beta}R_r+ \Theta(\frac{1}{n^2}) + \Theta(\frac{1}{n})=(x^*_i)^\beta n^{\beta}R_r+ \Theta(\frac{1}{n})&\Rightarrow\\
               (x^*_i)^\beta &= \frac{R_a+\Theta(\frac{1}{n})}{n^\beta R_r} = \frac{R_a(1+\Theta(\frac1n)}{n^\beta R_r} &\Rightarrow\\
               x^*_i&=\frac1n\left(\frac{R_a}{R_r}\right)^{\frac1\beta}(1+\Theta(\frac1n))^{\frac1\beta}
               = \frac1n\left(\frac{R_a}{R_r}\right)^{\frac1\beta}\left(1+\Theta(\frac1{n^{\frac1\beta}})\right)\\
               &= \frac1n\left(\frac{R_a}{R_r}\right)^{\frac1\beta} +\Theta\left(\frac1{n^{1+\frac1\beta}}\right),
    \end{align*}
    as required.
\end{proof}

\newpage
\setcounter{page}{1}
\begin{center}
    Online Appendices
\end{center}
\
\section{Details of the Experiment} \label{apx:experiment}
One-hundred and fifty four (81 Female) 
Technion and Ben Gurion University 
students, with an average age of 25, participated in the study in exchange for monetary compensation. The study included participants aged 18 and older who signed a consent form to participate in the experiment. The forms were signed by hand on a page in front of the research team and kept in the laboratory. The study was carried out between January 21, 2019 and April 28, 2019.
We planned to have at least 15 groups of 4 participants in each of the Protection and No Protection conditions, and stopped data collection once this goal was reached. Eventually we collected data from 60 participants in the Protection condition,\footnote{One student was mistakenly invited to the lab twice, and therefore her second session was removed.} 60 participants in the No Protection condition,
and 34 participants in the Singleton condition. 

A performance based payment was added to (if positive) or subtracted from (if negative) a show-up fee of 30 NIS. 
\footnote{Participants obtained a total of 29 NIS (that equals about \$8.3) on average, in a game lasting around 35 minutes. Note that the mean payoff is lower than the show-up fee, which means that on average, the performance based payment was negative. This is a first indication that participants did not behave optimally (they could guarantee the show-up fee by skipping all rounds).}

\subsubsection{Experimental Design}
Participants played a lab adaptation of the ``The Competitive Treasure Hunt" game that was described in Section~\ref{sec:golddigger}. 
\begin{enumerate}
  \item The hive included 2100 hexagons (70X30), rather than infinite number of hexagons. This modification implies that after each round, information is revealed and the probabilities change in the following ways: (1) the probability of finding a first treasure decreases after a treasure is found, since overall fewer treasures are left; (2) the probability of finding a first treasure increases after a failed search only for the player making the move, since players can observe all successful searches, but only their own failures. Consequently, the overall probability of finding a first treasure tends to decrease over time. However, since around each gold mine there are several known empty hexagons, the overall probability decreases only slightly.\footnote{We computed the probability of finding a first treasure in the last round of each game, and obtained on average the probabilities 0.0445, 0.0458 and 0.0492 in the Protection, No Protection and Singleton conditions, respectively. As previously described, the probability decreases on average over time, so these numbers estimate the minimum probability in each game. We can thus see that these probabilities are relatively close to 0.05.}
  \item Each game included 50 rounds, rather than infinite number of rounds. There is strong evidence that although in traditional game theoretic analysis any finite horizon may completely change the structure of equilibria, human players only take this into account (end-game effect) very close to the actual termination. E.g. in \citet{normann2012impact} end-game effect was explicitly measured only in the last 3 rounds out of 22 of the Repeated Prisoners' Dilemma. They also compared behavior under known, unknown, and random termination rules and find that differences in behavior only start $\sim$10 rounds before termination. 
  Moreover,  RPD is a deterministic game. Adding randomness to the game (as in our case) substantially reduces endgame effect, since it negates the value of looking ahead in general. E.g. while medium-level Chess programs typically consider $\sim$15 steps ahead, the Backgammon programs only need to look 2 steps ahead in order to beat the best human players \citep{tesauro1994td}.
  As we explain below, to compare participants' behavior to the theoretical, infinite time horizon benchmark, we excluded the last 12 rounds in each game from the analysis.
  \item Players were not informed in advance of the probability to find a treasure. As noted, with sufficient experience the learnt probability to find a treasure should converge to the actual one and lead rational players to a stable optimal threshold. The experiment included 4 games of 50 rounds each, which should allow for sufficient learning.\footnote{Indeed, analysis of potential changes throughout the games revealed quick learning and no significant differences between the first and the last game participants played.}
  \item If two or more players choose the same hexagon simultaneously, the payoffs that each player obtains follow this rule: if two players find the same treasure, each of them obtains 0.2 from the original reward of this treasure (which amount to 64 if this is the first treasure in the mine, and 16 if this is the second or the third). If three players find the same treasure, each of them obtains 0.05 from the original reward and if four players find it, each of them obtains 0.\footnote{This rule was designed to account for the fact that real life competition decreases the total producers surplus.} 
\end{enumerate}

\subsubsection{Procedure}
In each experiment's session, students invited to the lab were randomly assigned into groups of four. Each group was randomly assigned into the ``Protection" or the ``No Protection" condition. 
All the remaining participants were assigned to the ``Singleton" condition. 
Each participant played four games with 50 rounds per game.\footnote{18 participants from the Singleton condition played 5 rather than 4 games. For those participants, we analyzed only the first four games played.}

The players were not informed about the other players' identities, but they did know their group's size. 

The participants received three pages of instructions, which included pictures of different states of the game with explanation about the shape of mines and the meaning of different colors and marked areas (see Appendix~\ref{apdx:instructions}). Participants were informed about the structure of mines (i.e. a tight triangle), but not about their frequency and their location in the hive. In the Protection condition, the instructions explained that when a player finds the first treasure, he obtains an exclusive right to explore surrounding (adjacent) hexagons and benefit from the subsequent treasures. In the No Protection condition the instructions explained that the more players who find the same treasure, the lower the reward it yields. To make sure that the players understood the instructions well, before starting the first game they had to answer a short quiz with questions concerning the instructions of the game. The game started only after all the examinees answered all the questions correctly.
We did not mention any economic or domain specific terms such as ``protection", ``innovation" etc. in any of the condition's instructions.

Note that in all conditions, the players receive the same instructions (except for the introduction of protection rules in the protection condition), and the treasures location as well as the payoff procedure were the same.

 The game progressed as follows. First, the computer displayed the hive, containing 2100 hexagons. The players received a message stating the exploration cost for this round, and asking each player if he wants to skip the round or to explore under the current exploration cost. After making their choice, players were asked to wait for the other players to make their choices.\footnote{This message appeared also in the Singleton condition to match this condition to the other conditions.} 

If a player decided to skip, he gained 0, and the round was over for him. If a player decided to explore, he could choose one of the hexagons in the hive that was not yet colored. At the end of each round, the players received a message with their payoff from the round, calculated as the reward obtained minus the current exploration cost.

\subsection{Screenshots of the Competitive Treasure Hunt Game}\label{apdx:screenshots}
\begin{figure}[H]
    \centering
    \includegraphics[width = \textwidth]{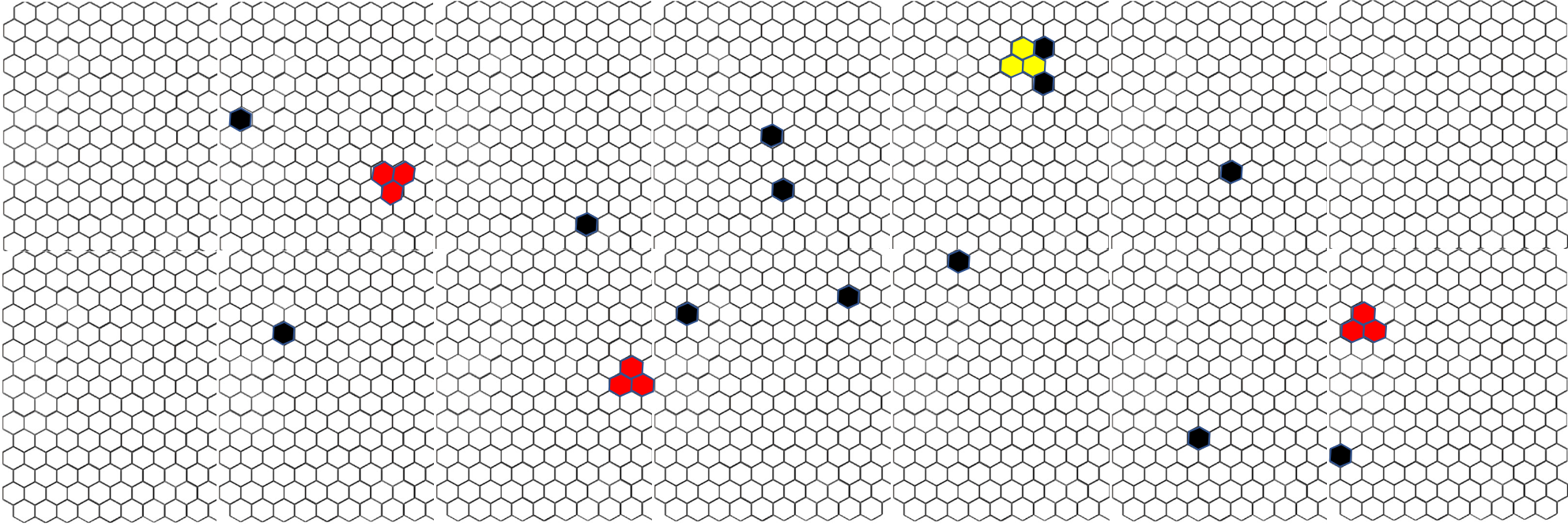}
    \caption{A screenshot of a typical end game under the protection condition.}
    \label{}
\end{figure}
\begin{figure}[H]
    \centering
    \includegraphics[width = \textwidth]{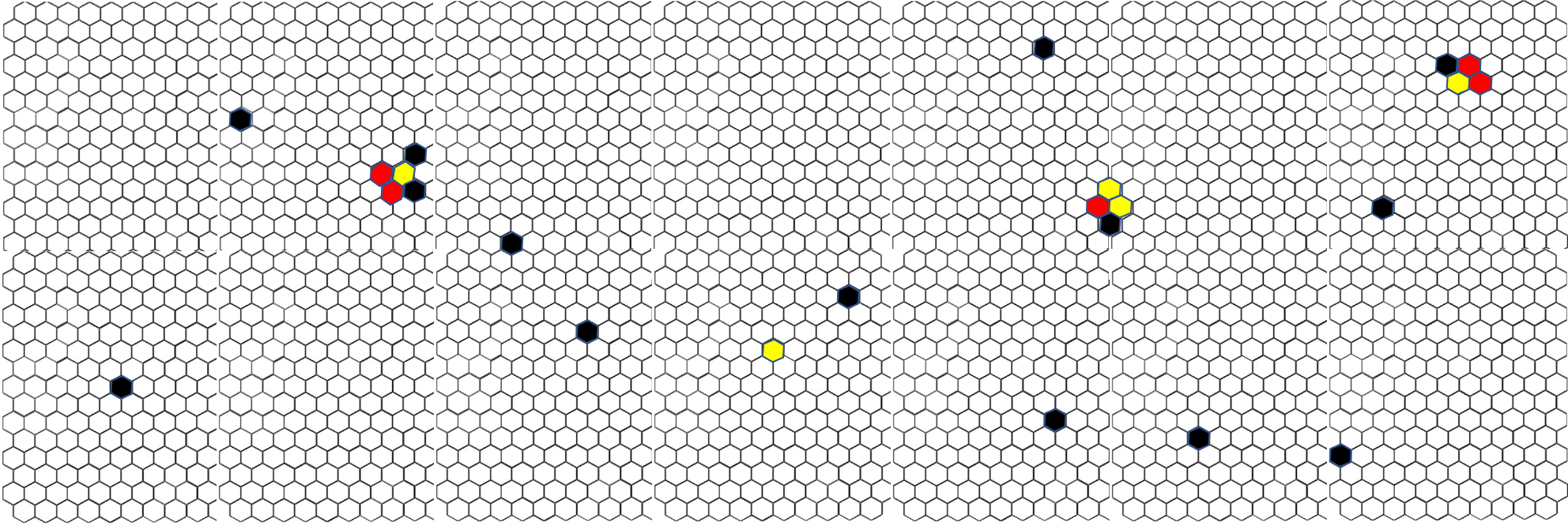}
    \caption{A screenshot of a typical end game under the no protection condition. }
    \label{}
\end{figure}

\section{Simulations}\label{apdx:simulation}
This section, presents the simulation results. The simulation obtained 7 different values for the cost threshold for first treasures and 7 different values for the cost threshold for subsequent treasures, overall 49 combinations of both. We ran the simulation 10000 times in any combination, and took mean values of the number of treasures and payoffs. 

Figures \ref{fig:individuallevel} and~\ref{fig:grouplevel} show that under a rational behaviour assumption (which is marked in a red circle),\footnote{In the No Protection condition, we took the cost within the possible range.} players find more treasures in the Protection condition, both at the group and at the individual levels. Figure \ref{fig:payoffsimulation} shows that in the Protection condition, optimal strategy leads to a payoff maximization, and in the No Protection condition, players could increase their payoffs by collectively deciding to explore less for subsequent treasures, below the equilibrium strategy. This result shows how lack of coordination among the players in the No Protection case, causes a reduction in payoffs.

In the simulation, we also measured the efficiency of exploration by the number of duplicated treasures, which is the amount of treasures that were found by more than one player. Figure~\ref{fig:duplicated} presents the number of duplicated treasures in the No Protection condition, and shows that this number increases as the thresholds increase. Structurally, there are no duplicated treasures in the Protection condition simulation.
\begin{figure}[H]
    \centering
    \includegraphics[width = \textwidth]{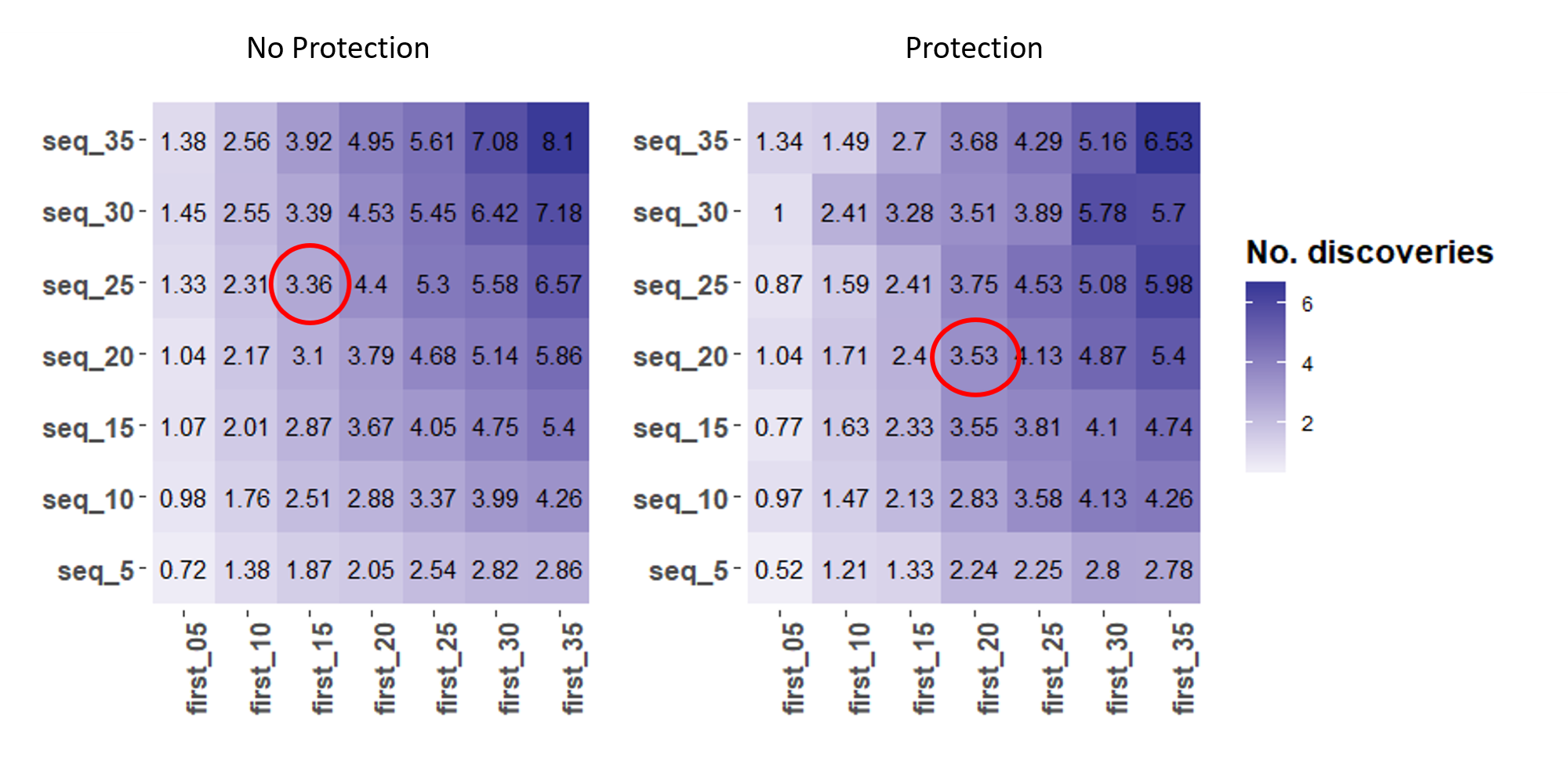}
    \caption{The number of treasures each player found, as a function of the chosen cost threshold in searching for first and subsequent treasures. In the No Protection condition we consider only symmetric strategies, where all players choose the same strategy. We can see that choosing the optimal strategies yields almost the same number of treasures in both conditions. }
    \label{fig:individuallevel}
\end{figure}
\begin{figure}[H]
    \centering
    \includegraphics[width = \textwidth]{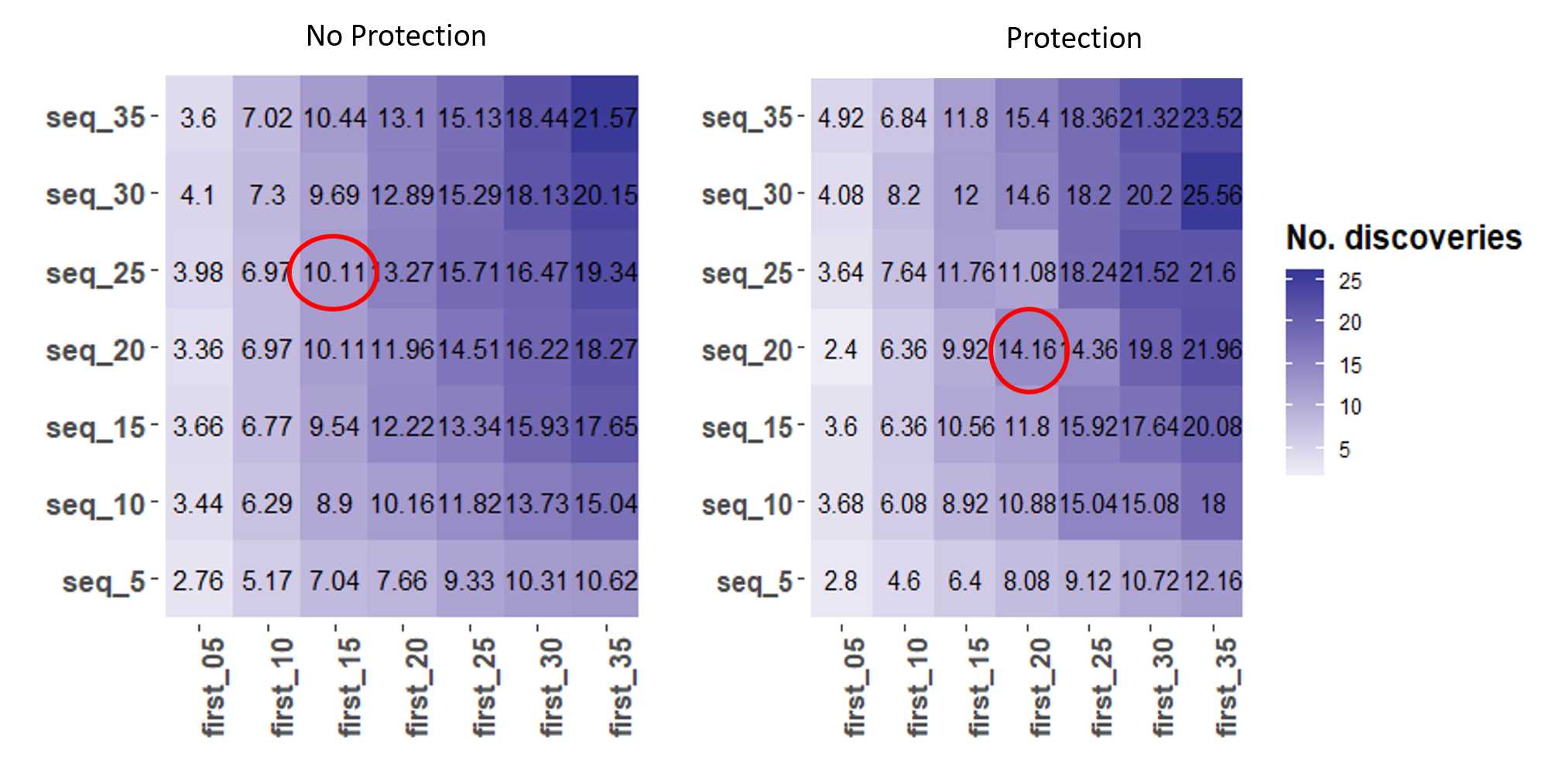}
    \caption{The number of treasures players found at the group level. In the No Protection condition we consider only symmetric strategies, where all players choose the same strategy. We can see that players found more treasures under the Protection condition.}
    \label{fig:grouplevel}
\end{figure}
\begin{figure}[H]
    \centering
    \includegraphics[width = \textwidth]{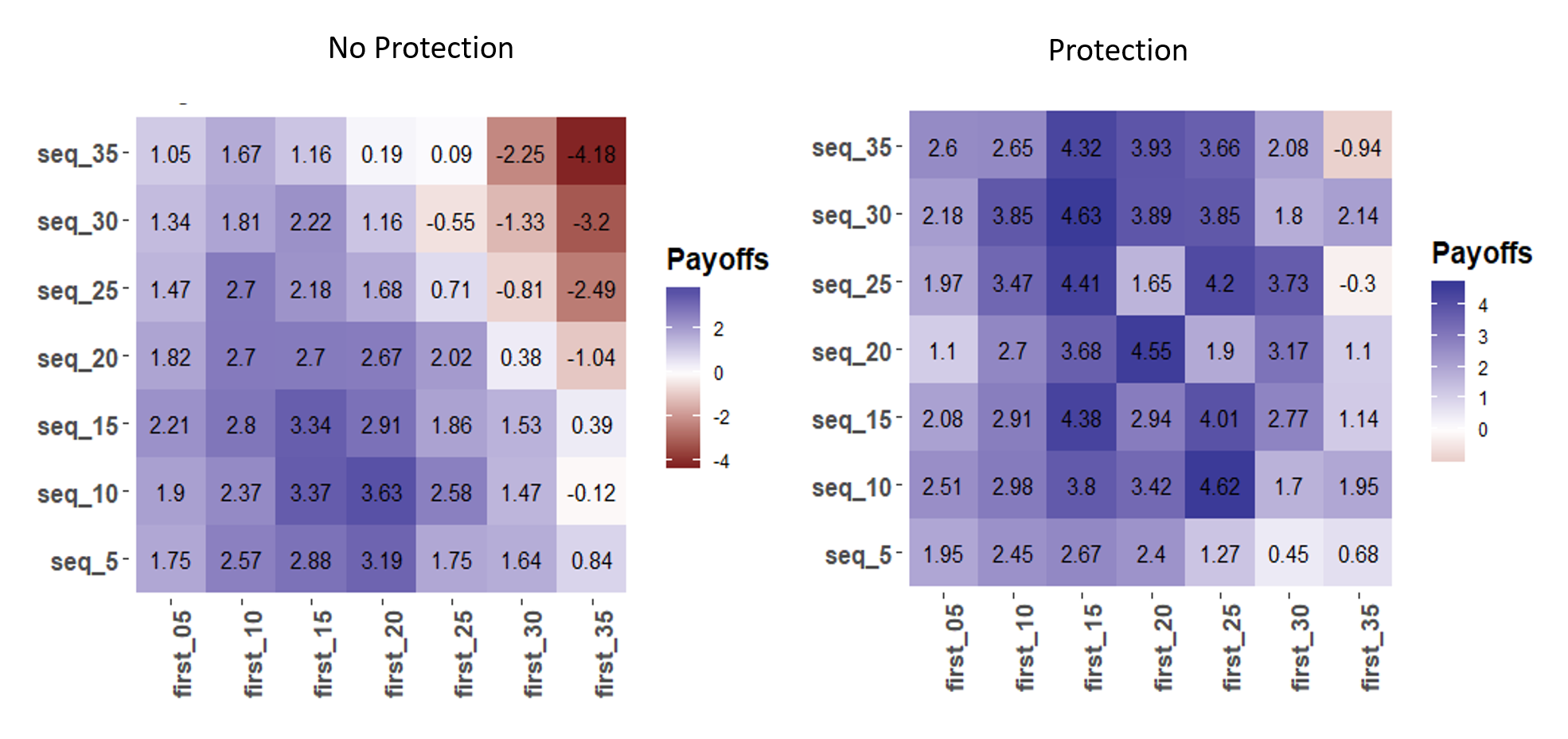}
    \caption{The payoffs as a function of cost threshold. In the No Protection condition we consider only symmetric strategies.}
    \label{fig:payoffsimulation}
\end{figure}
\begin{figure}[H]
    \centering
    \includegraphics[width = 0.5\textwidth]{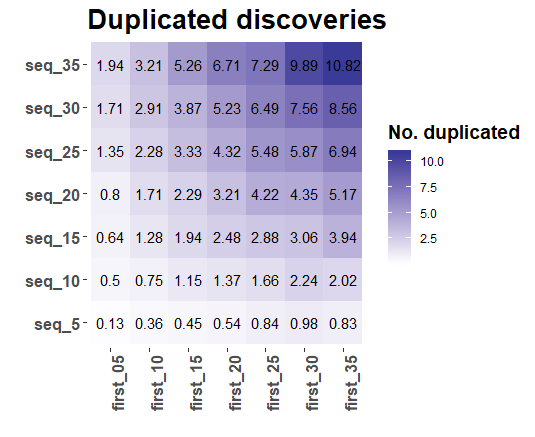}
    \caption{The number of treasures that were found by more than one player under the No Protection condition. }
    \label{fig:duplicated}
\end{figure}
\section{The Game Instructions}\label{apdx:instructions}
\begin{figure}[H]
\includegraphics[width=0.9\textwidth]{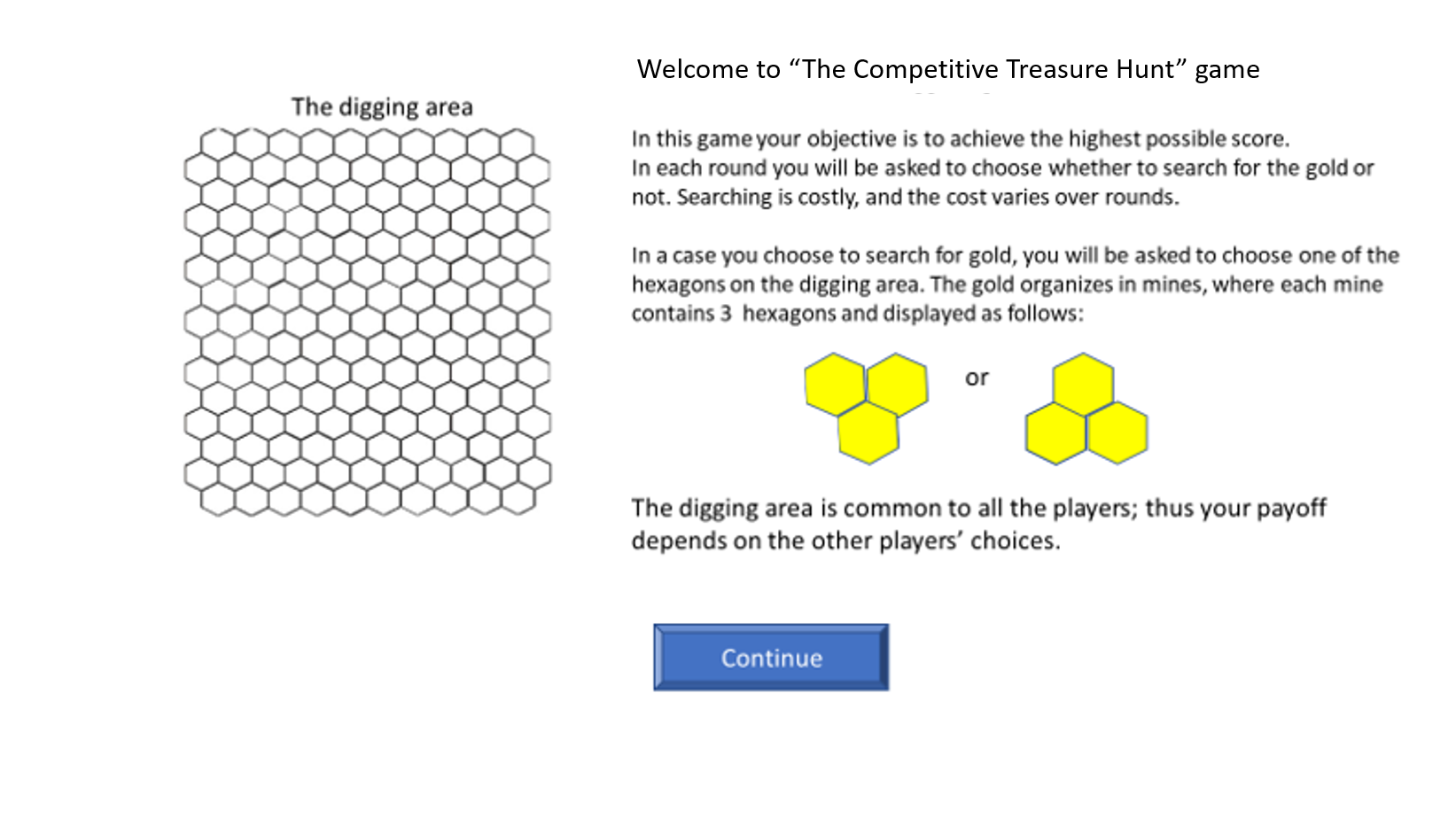}
\caption{Instructions that are common to all conditions}
\label{fig:Instruction1}
\end{figure}
\begin{figure}[H]
\includegraphics[width=0.9\textwidth]{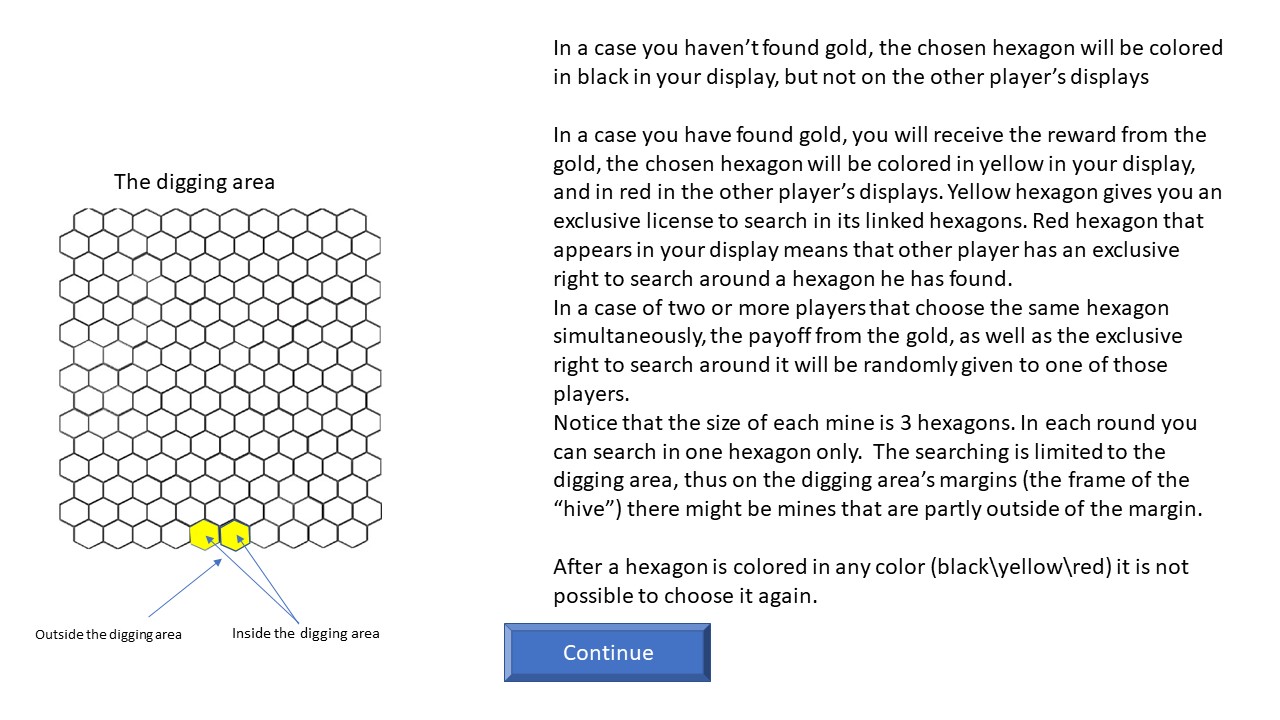}

\label{fig:Instruction3}
\end{figure}
\begin{figure}[H]

\includegraphics[width=0.9\textwidth]{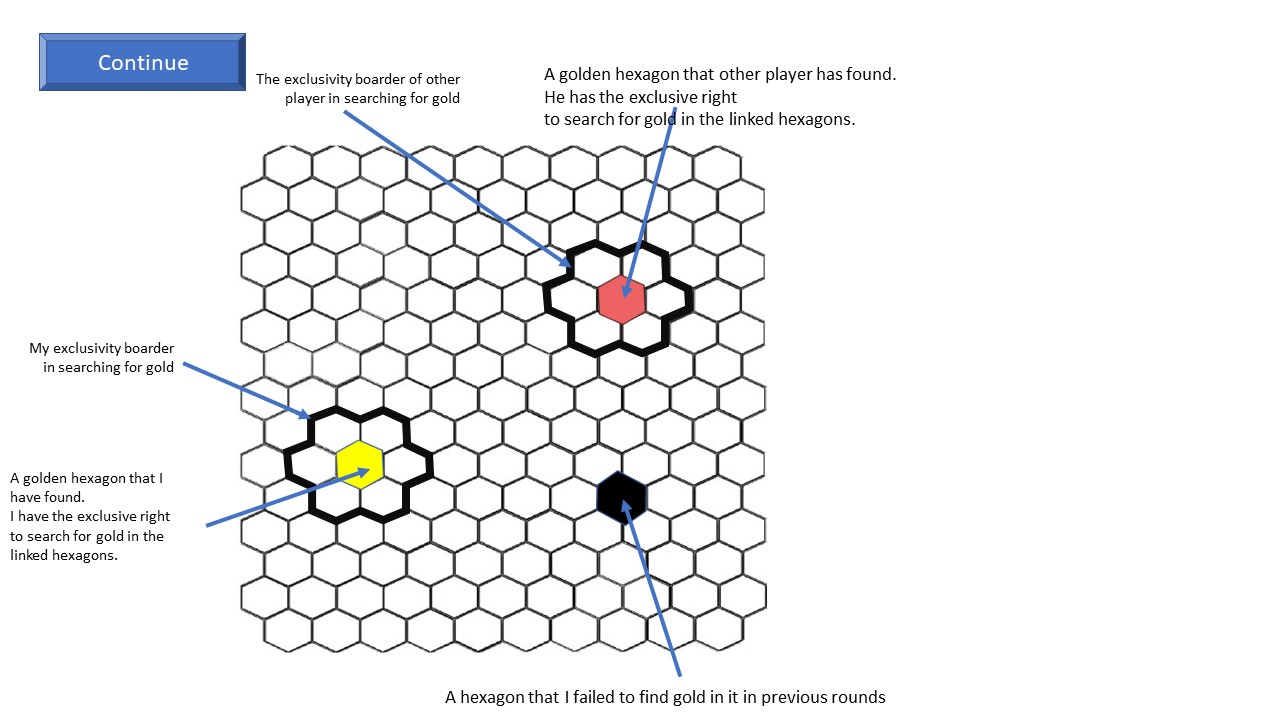}
\caption{Instructions to the Protection condition}
\label{fig:Instruction4}
\end{figure}

\begin{figure}[H]
\includegraphics[width=0.9\textwidth]{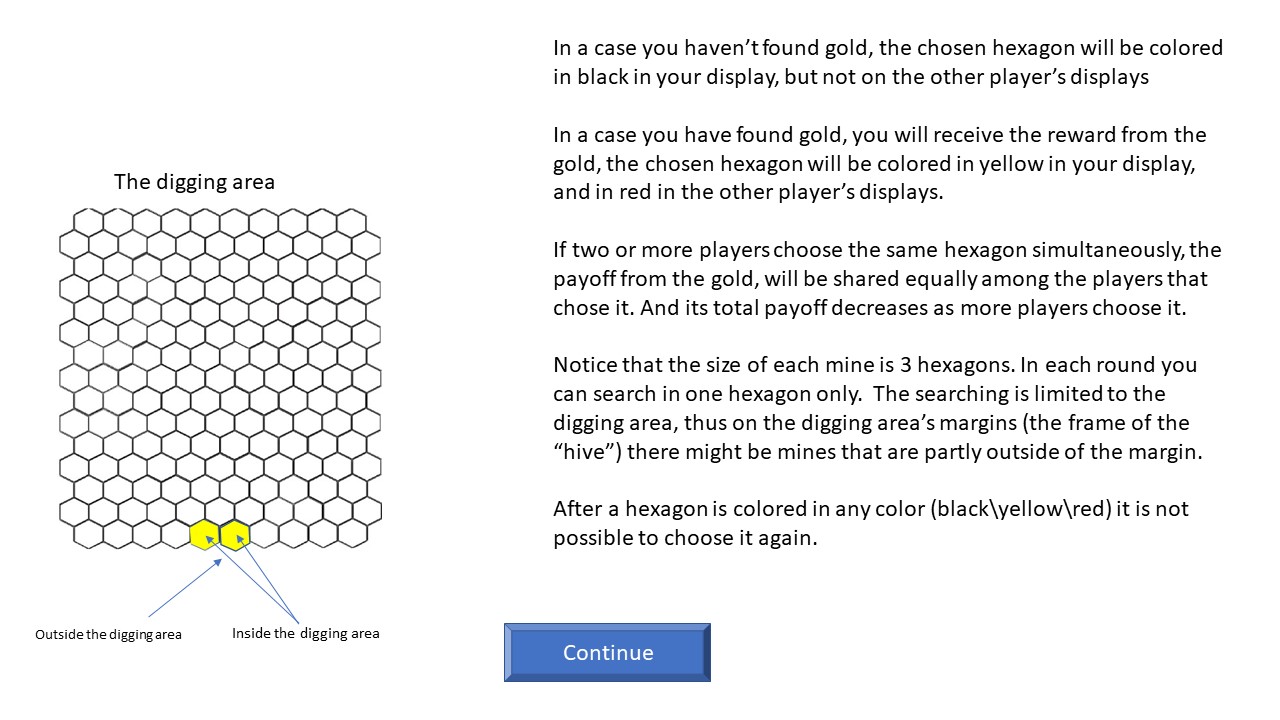}

\label{fig:Instruction5}
\end{figure}
\begin{figure}[H]

\includegraphics[width=0.9\textwidth]{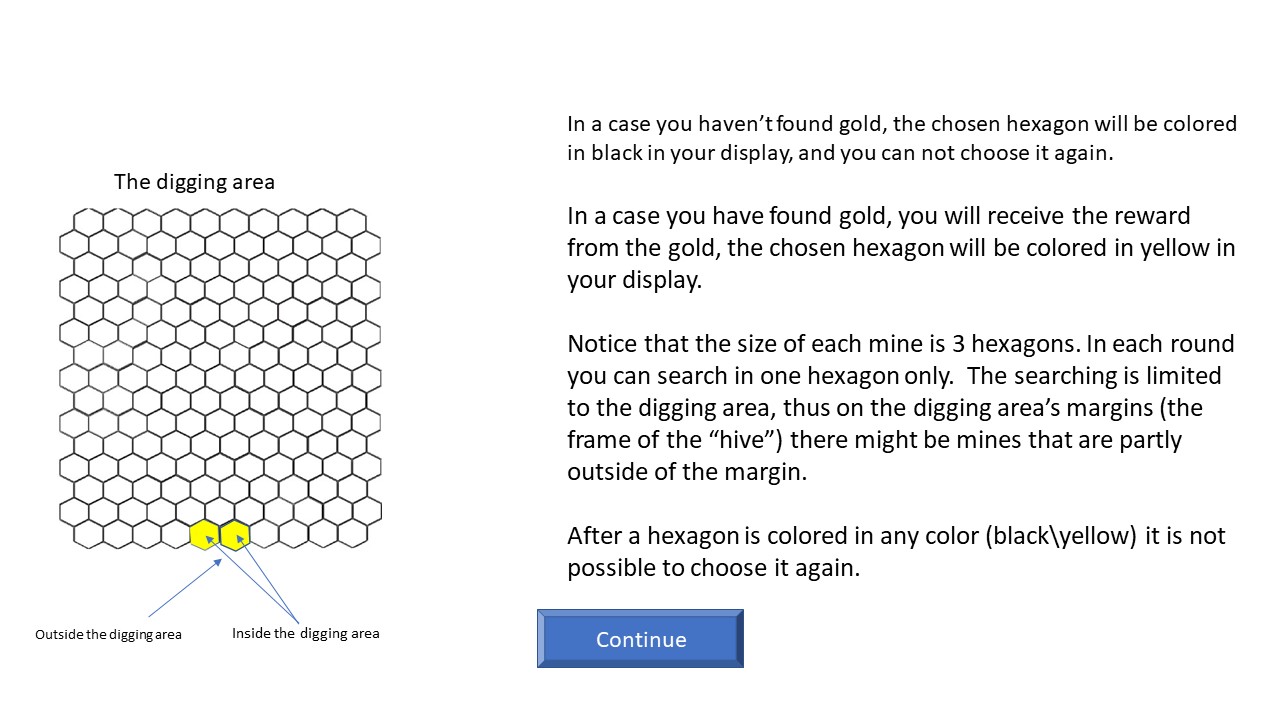}
\caption{Instructions to the No Protection condition}
\label{fig:Instruction6}
\end{figure}
\begin{figure}[H]
\includegraphics[width=0.9\textwidth]{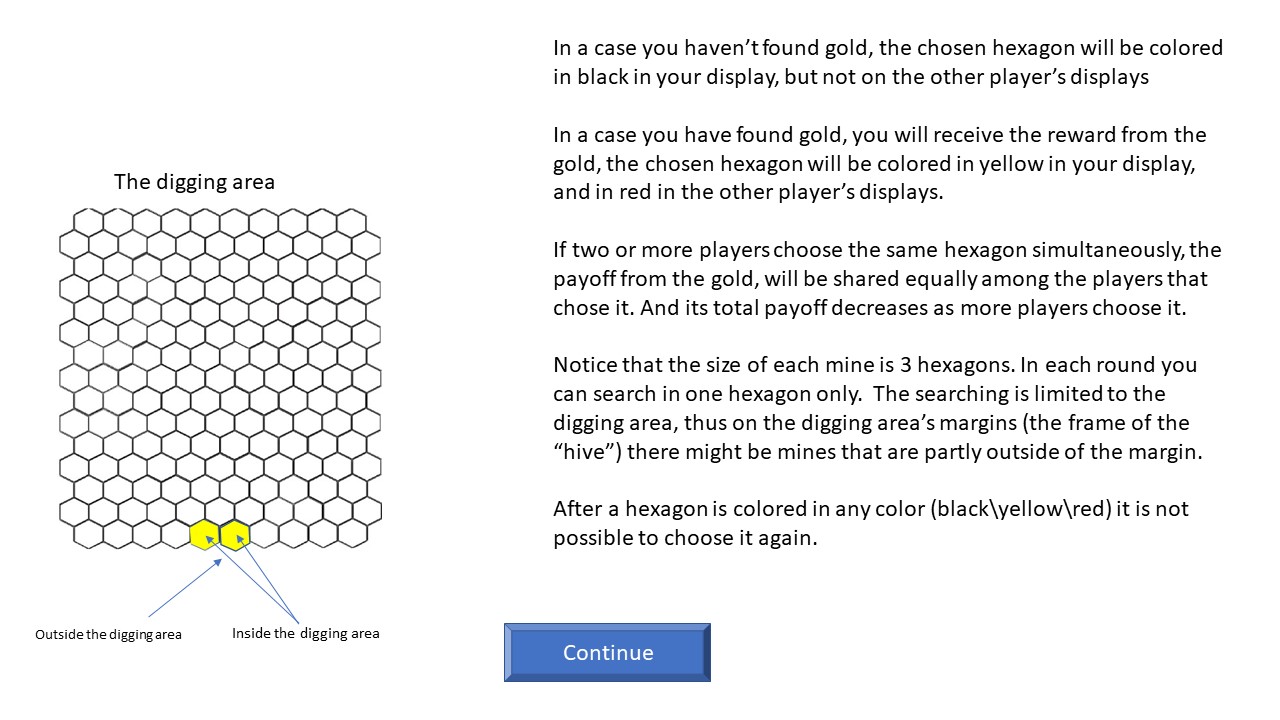}

\label{fig:Instruction7}
\end{figure}
\begin{figure}[H]

\includegraphics[width=0.9\textwidth]{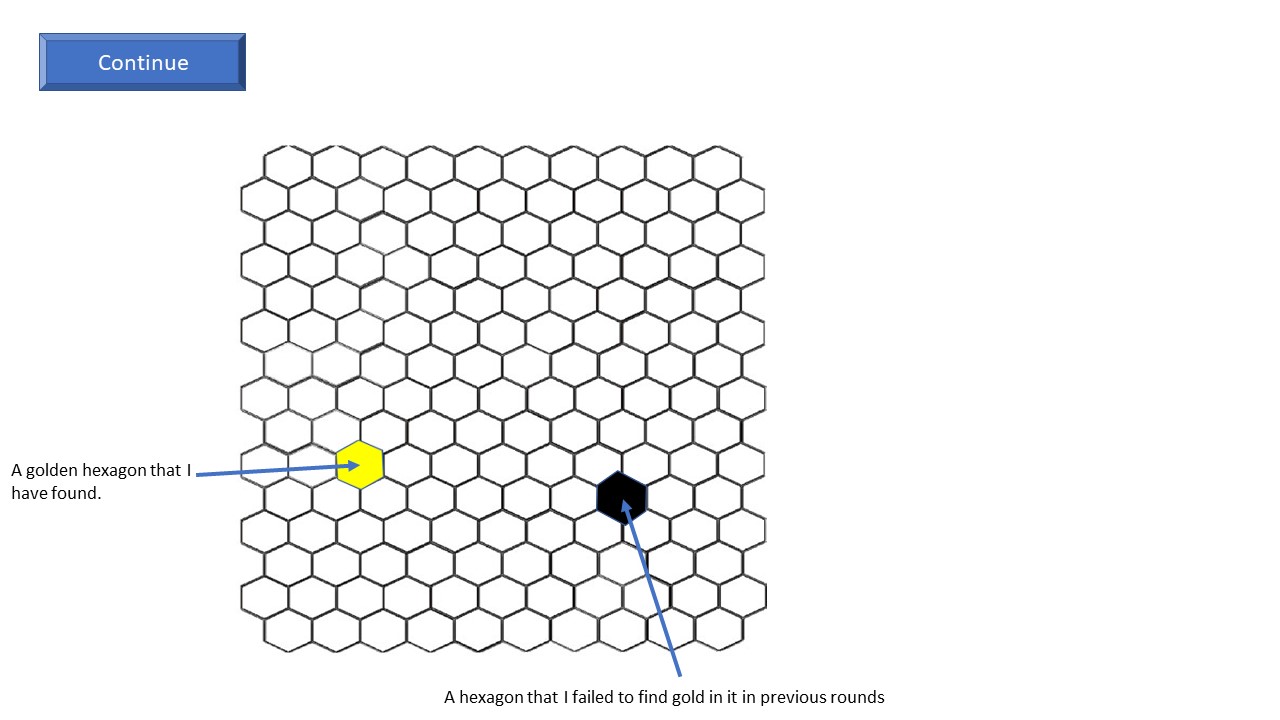}
\caption{Instructions to the Singleton condition}
\label{fig:Instruction8}
\end{figure}

\section{}\label{apdx:minepayoff}
\begin{figure}[H]
    \centering
    \includegraphics[width=0.6\textwidth]{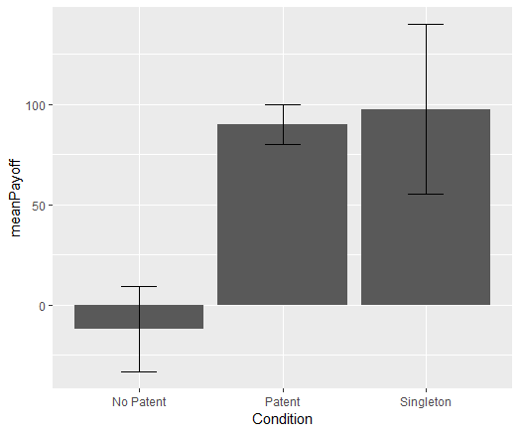}
    \caption{Average payoff from the second and
 the third treasure of any mine.}
    \label{fig:my_label}
\end{figure}

\section{}\label{apdx:levelsingleton}
\begin{figure}[H]
\includegraphics[width=0.6\textwidth]{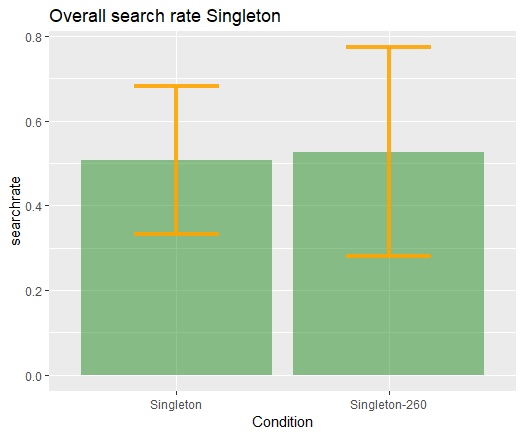}
\caption{Search rate for the first treasure in the Singleton condition when the reward from the first treasure equals 320 and 260, respectively. }
\label{fig:singleton}
\end{figure}

\section{Statistical analysis}\label{apx:stat}

In order to measure the correlation between participants and groups, we analyzed the data using a linear mixed effects model (LMM). We allow random intercepts for the players' ID and for the group of players.

\begin{table}[H]
\begin{center}
\begin{scriptsize}
\begin{tabular}{l c c c }
\hline
 & search & search & search \\
 & for first treasure & for self-subsequent treasure & for self-subsequent treasure\\
 & & & for high-costs only\\
\hline
(Intercept)                 & $1.11^{***}$  & $1.22^{***}$  & $1.45^{***}$  \\
                            & $(0.04)$      & $(0.05)$      & $(0.11)$      \\
Cost                           & $-0.03^{***}$ & $-0.02^{***}$ & $-0.02^{***}$ \\
                            & $(0.00)$      & $(0.00)$      & $(0.00)$      \\
Protection             & $0.06$        & $-0.10^{*}$   & $-0.17^{*}$   \\
                            & $(0.05)$      & $(0.05)$      & $(0.07)$      \\
\hline
AIC                         & 14734.27      & 273.42        & 305.32        \\
BIC                         & 14781.58      & 299.04        & 327.84        \\
Log Likelihood              & -7361.13      & -130.71       & -146.66       \\
Num. obs.                   & 19659         & 529           & 315           \\
Num. groups: ID             & 119           & 111           & 104           \\
Num. groups: GroupIndex     & 30            & 30            & 30            \\
Var: ID (Intercept)         & 0.06          & 0.02          & 0.06          \\
Var: GroupIndex (Intercept) & 0.00          & 0.01          & 0.01          \\
Var: Residual               & 0.12          & 0.08          & 0.10          \\
\hline
\multicolumn{4}{l}{\scriptsize{$^{***}p<0.001$, $^{**}p<0.01$, $^*p<0.05$}}
\end{tabular}
\end{scriptsize}
\caption{Search rate for first (column 1) and subsequent (column 2,3) treasures, column 3 presents search rate in high cost only.}
\label{table:Searching}
\end{center}
\end{table}

\paragraph{Initial search:} The first question addressed is how protecting first treasures affects exploration activity for initial search. Table~\ref{table:Searching}, column 1 presents the results of the exploration rates for first treasures. It shows that the coefficient of $Protection$ variable is insignificant, suggesting that when players explore for the first treasure in each mine, there is no significant difference between their behavior under the Protection and the No Protection conditions.

 Column~2 presents the results. It shows that protection on first treasures significantly decreases the overall tendency to explore for a subsequent treasure by 10 percent ($p<0.05$).

 Column~3 shows the same analysis for costs larger than 15. 

\section{Computing Thresholds}\label{apdx:actualthreshold}
In this section the methodology of computing the actual thresholds is presented.
We denote the states of the game by $F$ for exploration for a first treasure and $S$ for exploration for subsequent treasure. 
Notice that the definition for $S$ imposes asymmetric treatment between the Protection and the No Protection conditions. In the Protection condition, subsequent treasures are available only when the player finds the first treasure in the mine by himself. 
However, in the No Protection condition, subsequent treasures are available whenever any player finds the first treasure.\footnote{We removed from the analysis all observations in which we identified a search for a first discovery when it is possible to search for a subsequent discovery. As for the non-search classification when subsequent search is possible, it is in fact a non-search of both a first discovery and a subsequent discovery. It can therefore be classified as a non-search for a subsequent discovery, and so we did. Since there is no reason to believe that the threshold for a first search when a subsequent discovery is available, will be different from the threshold for a first search when it is not, removing the observations of these exceptional searches should not change the results.}
Thus, each round in the game is classified into these two states. For each player $i$, in each state of the game $\omega \in \{F , S\}$, we calculate the threshold cost value $T_{i,\omega}$ by the following process. 

First, for each possible cost $c_j \in \{5,10,15,20,25,30,35\}$ and for each round, $r$, and for each player $i$, we let $c_r^i$ be the realization of the cost for player $i$ in round $r$. We define the specification function: 
\begin{equation*}
S_i(c_j,c_r^i) = \begin{cases}
1 &\text{if $c_r^i \geq c_j$ and $search = 0$ or if $c_r^i < c_j$ and $search = 1$}\\
0 &\text{otherwise}
\end{cases}
\end{equation*} 
when ``search" is a variable that is set to be 1 when the player chooses to search, and to 0 when he decides to skip. Then for each $i,c_j,\omega$ we define the specification quality, which is essentially a 1-dimensional classifier with the 0-1 loss function, by: 
$$SQ(i,c_j,\omega) = \frac{\sum_{r\in \omega}S_{i}(c_j,c_r^i)}{|\omega_i|}$$ 
where $|\omega_i|$ is the number of rounds that player $i$ is in the state of the game $\omega$. Finally we define the threshold of each player in each state of the game as $T_{i,\omega}=argmax_{c_j}\{SQ(i,c_j,\omega)\}$ and the threshold quality by $TQ_{i,\omega} = max_{c_j}\{SQ(i,c_j,\omega)\}$ 

In other words, we consider each possible cost as a potential threshold. If this cost is the ``real" threshold, the player should explore whenever the cost realization is lower, and skip whenever it is higher. Then we took the potential threshold that yields the least number of mis-specifications of the ``real" threshold. Finally we define the threshold quality as the fraction of the number of rounds that the player's action was consistent with his threshold. 

Finally, we would like to find evidence for a consistent behaviour, that is expressed by a high quality of thresholds. Figure~\ref{fig:threshold_quality} represents the histogram of the thresholds quality. We can see that players tend to be highly consistent, where more than 84\% of the players had a threshold quality of 0.8 or more.

\begin{figure}[H]
\includegraphics[width=0.6\textwidth]{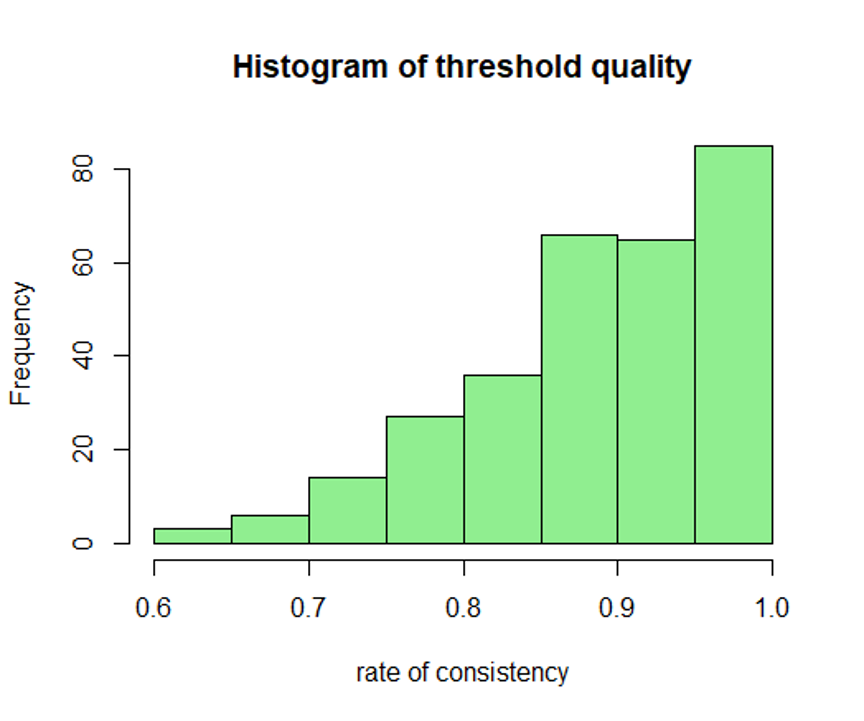}
\centering
\caption{Histogram of the threshold's quality}
\label{fig:threshold_quality}
\end{figure}

Fig.~\ref{fig:threshold} shows the distribution of thresholds we found in each condition and context.
\begin{figure}[H]
    \centering
    \includegraphics[width=0.7\textwidth]{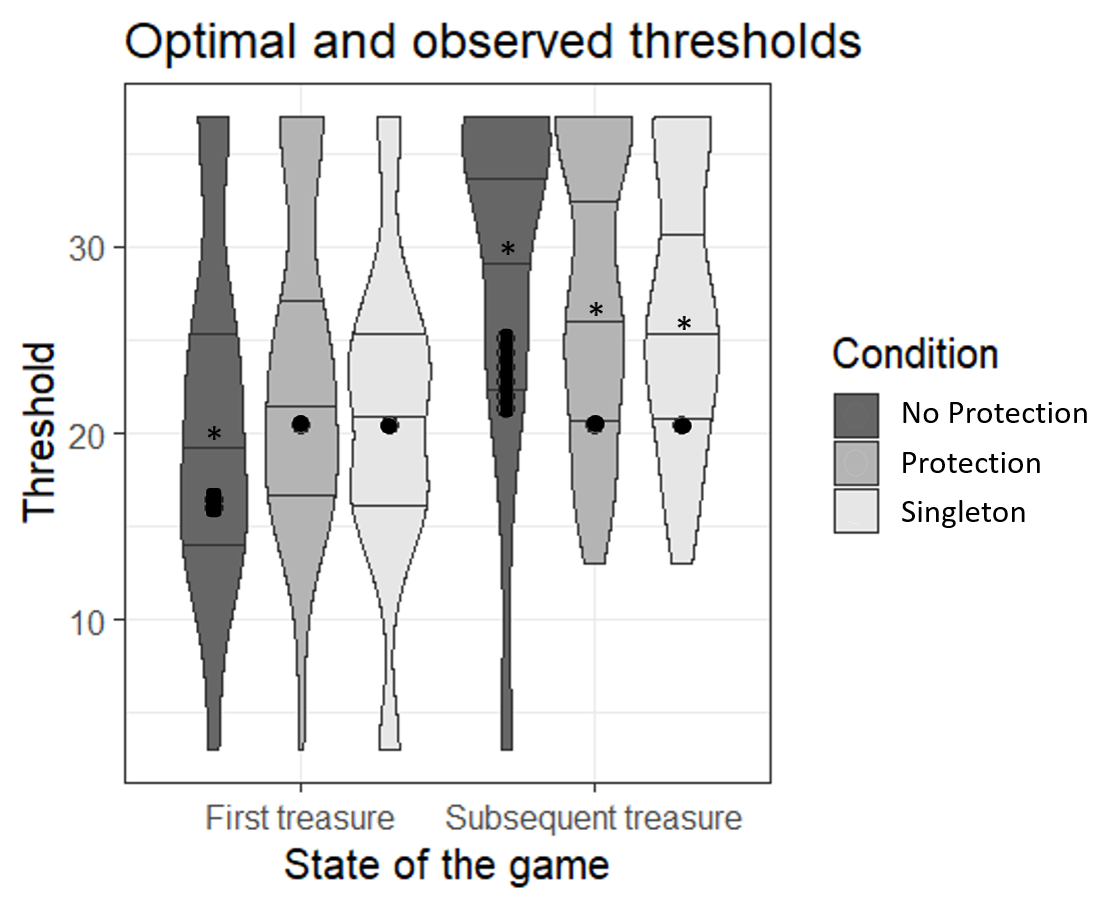}
    \caption{Threshold distribution by condition and by the state of the game. Dots represent optimal/equilibrium thresholds. Stars ($*$) represent significant differences between median value of observed threshold and optimal/equilibrium threshold, determined by the Wilcoxon test. Horizontal lines represent data deviation in quartiles.}
    \label{fig:threshold}
\end{figure}

\section{Protection vs. Singleton}\label{apdx:forgone}
We estimate the effect of the Protection condition on the exploration activity, compared to the Singleton condition, for cases where the players explore for a first treasure. 
Table~\ref{table:otherDiscoveries} column 1 presents the results of this estimation. We can see that under the Protection condition players explore 9 percentage points more than under the Singleton condition, however this effect is not significant. Although the coefficient of the condition is not precisely estimated, this result supports the hypothesis that revealing information about others' treasures encourages innovative activity of other inventors. 

Another way to show this hypothesis is to measure the effect of observing the other players' treasures within the Protection condition. We consider all observations in the Protection condition where players explore for a first treasure, and create a dummy variable, ``$other\_treasure$". 
This variable indicates whether there is another player that found a treasure in the previous round. We estimate the effect of this variable on the exploration activity in the current round. The result can be found in Table~\ref{table:otherDiscoveries} column 2. 
We can see that the players tend to explore 2 percentage points more for a first treasure after another player found treasure ($p<0.01$).
\begin{table}[H]
\begin{center}
\begin{tabular}{l c c }
\hline
 & search & search \\
\hline
(Intercept)                 & $1.12^{***}$  & $1.19^{***}$  \\
                            & $(0.04)$      & $(0.03)$      \\
$cost$                           & $-0.03^{***}$ & $-0.03^{***}$ \\
                            & $(0.00)$      & $(0.00)$      \\
$Protection$              & $0.09$        &               \\
                            & $(0.05)$      &               \\
$other\_treasure$            &               & $0.02^{**}$   \\
                            &               & $(0.01)$      \\
\hline
AIC                         & 13187.09      & 7528.65       \\
BIC                         & 13225.82      & 7571.92       \\
Log Likelihood              & -6588.54      & -3758.32      \\
Num. obs.                   & 17087         & 10015         \\
Num. groups: ID             & 94            & 59            \\
Var: ID (Intercept)         & 0.05          & 0.05          \\
Var: Residual               & 0.12          & 0.12          \\
Num. groups: GroupIndex     &               & 15            \\
Var: GroupIndex (Intercept) &               & 0.01          \\
\hline
\multicolumn{3}{l}{\scriptsize{$^{***}p<0.001$, $^{**}p<0.01$, $^*p<0.05$}}
\end{tabular}
\caption{The effect of observing the other players' treasures. Column 1 compares the Protection and the Singleton conditions and Column 2 estimates the effect of other players' successful search in the previous round, within the Protection condition.}
\label{table:otherDiscoveries}
\end{center}
\end{table}

\end{document}